\documentclass[11pt,a4paper]{article}
\usepackage[a4paper]{geometry}
\usepackage{amssymb,latexsym,amsmath,amsfonts,amsthm}
\usepackage{graphicx}
\usepackage{epstopdf}
\usepackage{tikz}
\usepackage{pgflibraryshapes}
\usetikzlibrary{arrows,decorations.markings}
\usepackage{subfigure}
\usepackage{overpic}
\usepackage{fullpage}
\usepackage{comment,verbatim}
\usepackage{hyperref}
\usepackage{color}
\usepackage{mathrsfs}
\usepackage[title,titletoc]{appendix}
\usepackage{ulem}
\usepackage{xcolor}

\renewcommand{\Re}{\mathrm{Re}\,}
\renewcommand{\Im}{\mathrm{Im}\,}

\newcommand{\ud}{\,\mathrm{d}}

\newcommand{\Boh}{\mathcal{O}}
\def\red{\color{red}}
\def\blue{\color{blue}}

\newtheorem{theorem}{Theorem}[section]
\newtheorem{lemma}[theorem]{Lemma}
\newtheorem{proposition}[theorem]{Proposition}

\newtheorem{rhp}[theorem]{RH problem}

\newtheorem{Riemann-Hilbert Problem}{Definition}

\theoremstyle{definition}

\theoremstyle{remark}

\newtheorem{remark}[theorem]{Remark}

\numberwithin{equation}{section}

\begin{document}

\title{On The Eigenvalue Rigidity of the Laguerre Unitary Ensemble}

\author{Chenhao LU\footnotemark[1] \ and Xiaolu YUE\footnotemark[2]}

\renewcommand{\thefootnote}{\fnsymbol{footnote}}
\footnotetext[1]{Department of Mathematics, City University of Hong Kong, Tat Chee
Avenue, Kowloon, Hong Kong. E-mail: \texttt{xiaolyue@ciyu.edu.hk}}
\footnotetext[2]{Department of Mathematics, City University of Hong Kong, Tat Chee
Avenue, Kowloon, Hong Kong. E-mail: \texttt{chenhaolu3-c@my.cityu.edu.hk}}

\date{\today}

\maketitle

\begin{abstract}
In this paper, we establish an optimal global rigidity estimate for the eigenvalues of the Laguerre unitary ensemble. Using the central limit theorem, we first construct a random measure via the eigenvalue counting function and then prove its convergence to a Gaussian multiplicative chaos measure, which yields the desired rigidity result. To prove this convergence, we apply a sufficient condition due to Claeys et al.~\cite{CFL2021} and carry out an asymptotic analysis of the corresponding exponential moments. 
\end{abstract}



 \tableofcontents

\section{Introduction}

\subsection{Main Results}

In this paper, we consider a generalized Laguerre Unitary Ensemble (LUE), whose density of the eigenvalues is given by
\begin{equation} \label{density of LUE}
\rho_N(\lambda_1, \lambda_2, \cdots, \lambda_N) = \frac{1}{Z_N} \prod_{1\leq \lambda_k<\lambda_j \leq N} |\lambda_k - \lambda_j|^2 \prod_{j=1}^N (1+\lambda_j)^\alpha e^{-NV(\lambda_j)}, \quad \lambda_j > -1
\end{equation}
where $Z_N$ is the normalization constant, and $V(x)$ in real-analytic on $[-1, \infty)$. We further assume that $V$ is one-cut and regular, with
\begin{equation}
    \lim_{x\to +\infty} \frac{V(x)}{\log x} = +\infty.
\end{equation}

The limiting eigenvalue distribution in this case is given by
\begin{equation} \label{LUE eq measure}
    d\mu_L = \psi_V(\lambda)\sqrt{\frac{1-\lambda}{1+\lambda}}d\lambda, \quad \lambda\in [-1, 1]
\end{equation}
where $\psi_V(x)$ is real-analytic on $[-1, \infty)$ and $\psi_V(x) > 0$ for $x \in [-1, 1]$. It implies that for any continuous $f$ with compact support on $[-1, +\infty)$, one has
\begin{equation}
\sum_{i=1}^N f(\lambda_i) \to \int_{-1}^\infty f(\lambda)d\mu_L
\end{equation}
in probability as $N\to\infty$. For more details, see \cite{CG2021, FPJ}. For the standard LUE where $V(x) = 2(x+1)$, we have $\psi_V(x) = \frac{1}{\pi}$. Here we use the same convention as that in \cite{CG2021} to make the support of the equilibrium measure be $[-1, 1]$, the same as the Gaussian unitary ensemble (GUE) and the Jacobi unitary ensemble(JUE); hence, we can make direct comparisons with the rigidity results on GUE in \cite{CFL2021} and the JUE in \cite{Dai-Lu-JUE}. As we may see below, our results show that the same optimal rigidity estimation holds for all three of these classic ensembles.

The corresponding distribution function is given as 
\begin{equation}\label{distribution of eq measure}
    F(x) = \int_{-1}^x d\mu_L.
\end{equation}
Define the $j$-th percentile $\kappa_j$ as
\begin{equation}\label{kappa}
\int_{-1}^{\kappa_j} d\mu_{L} = \frac{j}{N},
\end{equation}
which can be viewed as the theoretical limit of $\lambda_j$. We are particularly interested in how much $\lambda_j$ can fluctuate around $\kappa_j$. This kind of problem is called eigenvalue rigidity.


The eigenvalue rigidity of general Wigner matrices was established in a foundational series of works by by Erd\"{o}s, Yau et al. in a \cite{Erd:Sch:Yau2009-cmp, Erd:Sch:Yau2009, Erd:Yau:Yin2012-ptrf, Erd:Yau:Yin2012}. As a canonical case of complex Wigner matrices, the eigenvalue density of the classical GUE is given by
\begin{equation} \label{eq: GUE density}
\rho_N(\lambda_1, \lambda_2, \cdots, \lambda_N) = \frac{1}{Z_N} \prod_{1\leq \lambda_k<\lambda_j \leq N} |\lambda_k - \lambda_j|^2 \prod_{j=1}^N  e^{-2N \lambda_j^2}.
\end{equation}
From \cite[Thm. 2.2]{Erd:Yau:Yin2012}, one can see that for suitable constants $\alpha > \alpha' > 0$ and $C, c > 0$ 
\begin{equation}
\mathbb{P}\left(  \max_{j=1, \ldots, N} \left\{ \sqrt{1-\kappa_j^2} \, |\lambda_ k - \kappa_j| \right\} \geq \frac{(\log N)^{\alpha \log \log N}}{N} \right) \leq C \exp \left( -c (\log N)^{\alpha' \log \log N} \right).
\end{equation}
As a consequence, with high probability, the global eigenvalue fluctuations of a GUE matrix are bounded by $(\log N)^{\alpha \log \log N}/N$.

However, it is difficult to obtain the lower bound for the eigenvalue fluctuations; hence, one cannot decide whether the rigidity results are optimal. The first optimal rigidity result for random matrix models was obtained for the circular unitary ensemble, where the eigenvalues are distributed on the unit circle. Consider the ordered eigenangles $0\leq \theta_1\leq \cdots \leq \theta_N < 2\pi$ with density
\begin{equation}
\rho_N(\theta_1, \cdots, \theta_N) = \frac{1}{Z_N} \prod_{1\leq \theta_k < \theta_j \leq N}|e^{i\theta_k} - e^{i\theta_j}|^2.
\end{equation}
It is known that the limiting distribution of CUE is uniform on the unit circle; for instance, see \cite{FPJ}. 
The $j$-th eigenangle is classically located at $2 \pi j/N$. As shown by Arguin, Belius, and Bourgade \cite[Theorem 1.5]{ABB2017},  for any $\varepsilon>0$, one has
\begin{equation} \label{eq:cue-optimal}
\lim_{N \to \infty}
\mathbb{P}\left( (2-\varepsilon)\frac{\log N}{N} \leq \max_{j=1, \cdots, N} \left| \theta_j-\frac{2\pi j}{N}\right| \leq (2+\varepsilon)\frac{\log N}{N} \right) = 1.
\end{equation}
Similar results were also obtained in \cite{Paq:Zei2018}. This implies that the maximum fluctuation of the eigenangles $\theta_j$ around their classical location $2 \pi j/N$ is of order $O(\log N / N)$. The difference between the coefficients of the upper bound and lower bound can be arbitrarily small. Hence, we can claim that this estimation must be optimal.

It was only with the introduction of Gaussian multiplicative chaos (GMC) that more optimal estimations arising from random matrix theory were obtained, and we will discuss them in detail in Section \ref{subsubsec: GMC}.

Define the eigenvalue counting function as
\begin{equation}\label{definition of hN}
    h_N(x) = \sqrt{2}\pi \left(\sum_{1\leq j \leq N} 1_{\lambda_j \leq x} - N\int_{-1}^x d\mu_L\right).
\end{equation}
Then it follows directly that
\begin{equation}
    h_N(\lambda_j) = \sqrt{2}\pi N(F(\kappa_j) - F(\lambda_j)),
\end{equation}
where $d\mu_L$ and $F(\cdot)$ are defined in \eqref{LUE eq measure} and \eqref{distribution of eq measure}.

Our first result is regarding the maximum of $h_N(x)$. To be specific, we have the following theorem.
\begin{theorem} \label{maximum of hN}
For $h_N$ defined in \eqref{definition of hN}, we have
\begin{equation}
    \lim_{N\to\infty} \mathbb{P}\left(\sqrt{2}(1-\varepsilon)\log N \leq \max_{x\in[-1, \infty)}h_N(x) \leq \sqrt{2}(1+\varepsilon)\log N\right) = 1
\end{equation}
for any $\varepsilon > 0$.
\end{theorem}
It is worth mentioning that this is the optimal estimation for the maximum of $h_N(x)$ on $[-1, \infty)$, since the upper bound and lower bound differ only by a sufficiently small constant.

Our next result concerns eigenvalue rigidity, i.e., the estimation of $\max_jF'(\kappa_j)|\lambda_j - \kappa_j|$. For $j$ not close to $-1$ or $1$, a simple Taylor expansion shows that it can be well approximated by $h_N(x)$. But near the edge, especially near $-1$, the term $F'(\kappa_j)$ explodes. Hence, further analytical techniques are necessary, applied in tandem with Theorem \ref{maximum of hN}.
\begin{theorem} \label{main theorem}
Let $\lambda_1\leq\cdots\leq \lambda_N$ be distributed with respect to \eqref{density of LUE}. Then for any $\varepsilon>0$, we have
\begin{equation} \label{eq:main-theorem}
\lim_{N\to\infty}\mathbb{P}\left((1-\varepsilon)\frac{\log N}{N} < \max_{j=1, \cdots, N}\left\{F'(\kappa_j)|\lambda_j-\kappa_j|\right\} < (1+\varepsilon)\frac{\log N}{N}\right) = 1,
\end{equation}
where $\kappa_j$ is defined in \eqref{kappa} and $F(x)$ is defined in \eqref{distribution of eq measure}.
\end{theorem}
For the same reason, we can see that the estimation of the maximum fluctuations of eigenvalues is optimal. And as we have mentioned above, this theorem is not a direct corollary of Theorem \ref{maximum of hN}, since there $-1$ is a singularity of $F'(x)$, which makes $F'(x)$ increase very quickly as $\kappa_j$ tends to $-1$, or for $j$ very small. We will elaborate on our method of addressing this issue in Section \ref{subsubsec: refinement}.

One can also see that the order of rigidity is $O(\log N/ N)$ in the case of LUE, which is the same as CUE in \cite{ABB2017}, GUE in \cite{CFL2021}, and JUE in \cite{Dai-Lu-JUE}. More research showed that such a bound holds for other ensembles, such as Wigner matrices and Gaussian $\beta$ ensemble (G$\beta$E); for instance, see \cite{BLZ-GAFA2025}. Thus, we believe that this result is universal in random matrices and holds uniformly for a wider class of ensembles.

\subsection{Main Methodology}

To prove our main theorem, i.e., Theorem \ref{main theorem}, we mainly follow the general framework given by \cite{CFL2021}. The main idea is to first adopt the central limit theorem for LUE, which implies that $h_N(x)$ defined in \eqref{definition of hN} converges to a log-correlated field. Then we prove that we can use $h_N$ to construct a Gaussian multiplicative chaos measure (GMC). To accomplish this, we need to verify that $h_N(x)$ satisfies the sufficient conditions given in \cite[Assumption 2.5]{CFL2021}. In this part, we need to make some detailed estimations on the asymptotics of the exponential moments of $h_N$, for which we used Heine's identity to transfer it into the estimation of the related Hankel determinants and then adopted Riemann-Hilbert (RH) analysis to perform the asymptotic analysis. Combining with the results in \cite[Sect. 4]{CG2021}, it implies Theorem \ref{maximum of hN}. In the last stage, we use the method of iteration to ``refine'' the estimation of $h_N(x)$ near $\pm 1$, which concludes the final step of the proof of Theorem \ref{main theorem}. It is worth mentioning that our method of refinement can be adopted in both soft edge $1$ and hard edge
$-1$, and is expected to extend to other ensembles with more general eigenvalue limiting distributions.
\subsubsection{Log-Correlated Fields and Gaussian Multiplicative Chaos} \label{subsubsec: GMC}

We start with a central limit theorem for LUE, which is illustrated in \cite[Corollary 2.2]{CG2021}: let $f: [-1, \infty) \to \mathbb{R}$ be an analytic function in a neighbourhood of $[-1, 1]$, locally H\"{o}lder-continuous, then we have
\begin{equation} \label{eq: LUE-Global-CLT}
-\frac{1}{\sqrt{2} \, \pi}\int_{-1}^{1} f'(x) h_N(x) dx  =
\sum_{j=1}^N f(\lambda_j) - N\int_{-1}^1 f(x) d\mu_L(x) \xrightarrow{\ d \ } \mathcal{N}(\mu(f), \sigma^2(f)),
\end{equation}
where 
\begin{eqnarray} \label{eq: LUE-Global-CLTmean}
\mu(f) &=& \frac{\alpha}{2\pi} \int_{-1}^{1}\frac{f(x)}{\sqrt{1-x^2}}dx - \frac{\alpha}{2}f(-1), 
\end{eqnarray}
and $\sigma^2(f) = \sigma^2(f; f)$ with
\begin{eqnarray}
\sigma^2(f; g) &=& \int\int_{ [-1, 1]^2} f'(x)g'(y)\frac{\Sigma(x, y)}{2\pi^2}dxdy, \\
\Sigma(x, y) &=& \log\left|\frac{1-xy+\sqrt{1-x^2}\sqrt{1-y^2}}{x-y}\right|. \label{correlation kernel of X}
\end{eqnarray}
Hence, $h_N(x)$ converges in distribution to a log-correlated Gaussian field, denoted as $X(x)$, with correlation kernel \eqref{correlation kernel of X}. 

Note that the variance of $X(x)$ is infinite; hence, for each $x$, $X(x)$ should be viewed as a generalized random variable. Such kinds of random fields occur in many disciplines. For a formal definition and more properties, see \cite{DRSV2017}. Starting from the pioneering work by \cite{Johansson1998}, where the log-correlated covariance structure for the Gaussian $\beta$ ensemble was established, the log-correlated fields arising from the central limit theorems for the eigenvalues of random matrices have been discovered in many other models. For instance, in \cite{DP2012} Dumitriu and Paquette studied the generalized Jacobi $\beta$ ensembles. And in \cite{BB-PAFA-2021, CG2021}, a summary of the central limit theorems for all three classic ensembles (GUE, JUE, and LUE) was given.

The study of the GMC measure can be traced back to \cite{Kahane1985} in 1985, where Kahane aimed to study the exponential of a random field with logarithmic singularities. This problem is nontrivial since the exponential of a generalized function is not guaranteed to exist. He proved that for $0 < \gamma < \sqrt{2}$, the measure
\begin{equation}
d\mu^\gamma = \frac{e^{\gamma X(x)}}{\mathbb{E}e^{\gamma X(x)}}dx
\end{equation}
exists and is nontrivial.

After this pioneering work, the GMC measure has been found to have many nice properties and is of great significance in many disciplines, even beyond mathematics; for instance, in Liouville quantum gravity theory \cite{Rhodes2016}. See also \cite{Rho:Var2014} for a comprehensive review on this topic.

Recently, the idea of GMC has been introduced into the field of random matrix theory to obtain many significant results, especially those related to the estimation of both upper and lower bounds arising from random matrices. For instance, in \cite{CFL2021}, Claeys, Fahs, Lambert, and Webb built a general framework on how GMC measures obtained from asymptotic Gaussian processes are related to the maximum of the corresponding log-correlated fields; hence, it is proved that for GUE, one has
\begin{equation}
\lim_{N\to\infty}\mathbb{P}\left((1-\varepsilon)\frac{\log N}{N} < \max_{j=1, \cdots, N}\left\{F'(\kappa_j)|\lambda_j-\kappa_j|\right\} < (1+\varepsilon)\frac{\log N}{N}\right) = 1,
\end{equation}
with $\kappa_j$ being the $j$-th percentile of the semicircle law and $\lambda_1, \cdots, \lambda_N$ distributed with respect to \eqref{eq: GUE density}. Here, the prime denotes differentiation with respect to 
$x$. The similar idea has been adopted in \cite{Dai-Lu-JUE} to prove the eigenvalue rigidity of the Jacobi unitary ensemble.

With the powerful tool of GMC, many optimal estimations of eigenvalue fluctuations and characteristic polynomials in random matrices have been obtained, and it is worth mentioning that such results are not limited to the unitary ensembles. For instance, see \cite{Bre:Web:Wong2018, BLZ-GAFA2025, KW2022, Kivimae2024, Lam:Ostr:Simm2018, Lambert2020, NSW-CUE-2020, Webb-CUE-2015} and so on.

In our paper, we follow the ideas in \cite{CFL2021} by constructing a random measure 
\begin{equation} \label{eq:dmu-N}
d\mu_N^\gamma(x) = \frac{e^{\gamma h_N(x)}}{\mathbb{E} [e^{\gamma h_N(x)} ] }, \qquad x \in [-1,1],
\end{equation}
and prove its convergence to a GMC measure as $N \to \infty$. Then we can adopt their results in \cite[Section 3]{CG2021} to obtain Theorem \ref{maximum of hN}, the estimation of $h_N(x)$. As we have mentioned before, it requires us to study the asymptotics of the exponential moments of $h_N(x)$ in detail.

\subsubsection{Asymptotics of the Hankel Determinants}

We are now at the stage of discussing the relationship between the exponential moments of $h_N(x)$ and the Hankel determinants by considering a very simple example, i.e., $\mathbb{E} e^{\gamma h_N(x)}$. From the definition of $h_N(x)$ in \eqref{definition of hN}, one has
\begin{equation}
\mathbb{E} [e^{\gamma h_N(x)}] = \mathbb{E} \big[e^{\sqrt{2}\pi\gamma \sum_{j=1}^N 1_{\lambda_j\leq x} - \sqrt{2}\pi\gamma N F(x)}\big] 
= e^{-\sqrt{2}\pi \gamma N F(x)} \mathbb{E} \big[e^{\sqrt{2}\pi\gamma \sum_{j=1}^N 1_{\lambda_j\leq x}}\big],
\end{equation}
where $F(x)$ is the distribution function defined in \eqref{distribution of eq measure}.

Recalling Heine's identity (see, for example, \cite[Prop 3.8]{Deift2000}), we have
\begin{equation}\label{Heine identity}
 \mathbb{E}e^{\sum_{j=1}^N w(\lambda_j)} = \frac{N!}{Z_N} D_N(e^{w}) =\frac{N!}{Z_N} \det \left(\int \lambda^{i+j} (1+\lambda)^\alpha e^{-NV(\lambda)} e^{w}d\lambda\right)_{i,j=0}^{N-1},
\end{equation}
with $\lambda_1, \cdots, \lambda_N$ distributed with respect to \eqref{density of LUE}. Here $D_N(f)$ is defined as
\begin{equation}
D_N(f) = \det \left( \int_{[-1, \infty)}\lambda^{i+j}f(\lambda)(1+\lambda)^\alpha e^{-NV(\lambda)} d\lambda\right)_{i, j=0}^{N-1}.
\end{equation}
Hence, it motivates us to define the Hankel determinant
\begin{equation}
D_N(x; \gamma; w) =\det\left(\int_{[-1, \infty)} \lambda^{i+j} e^{\sqrt{2}\pi\gamma 1_{\{\lambda\leq x\}}} (1+\lambda)^\alpha e^{-NV(\lambda)} e^{w(x)} d\lambda\right)_{i,j=0}^{N-1},\label{eq2}
\end{equation}
and $\mathbb{E} e^{\gamma h_N(x)}$ can be written as 
\begin{equation}
\mathbb{E} [e^{\gamma h_N(x)}] = e^{-\sqrt{2}\pi \gamma N F(x)} \frac{N!}{Z_N} D_N (x; \gamma; 0) = e^{-\sqrt{2}\pi \gamma N F(x)} \frac{D_N (x; \gamma; 0)}{D_N (0; 0; 0)}.  \label{EV-Heine Identity1}
\end{equation}
Note that the exponential moments of $h_N(x)$ introduce a Fisher-Hartwig (FH) singularity in the Hankel determinant. For our purposes, we will also need the asymptotics of the Hankel determinants with two FH singularities; hence, we define
\begin{align}
D_N(x_1, x_2; \gamma_1, \gamma_2; w) = \det\left(\int_{[-1, \infty)} \lambda^{i+j} e^{\sqrt{2}\pi\gamma_1 1_{\{\lambda\leq x_1\}} + \sqrt{2}\pi\gamma_2 1_{\{\lambda\leq x_2\}}} (1+\lambda)^\alpha e^{-NV(\lambda)}  e^{w(x)} d\lambda \right)_{i,j=0}^{N-1}, \label{eq1}
\end{align}
where $-1 < x_1 \leq x_2 < 1$, $\gamma_1, \gamma_2 \in \mathbb{R}$, and $w(\lambda)$ are real-analytic in $\mathbb{R}$.

One can see that if we consider the exponential moments of an analytic function $f(x)$ rather than $h_N(x)$, it will only contribute to the $w$-term in the Hankel determinants rather than the FH singularities.

In the next section, it will be shown that, with the aid of the differential identities \eqref{di1} and \eqref{di2}, the estimation of the Hankel determinants can be reduced to that of the corresponding Riemann-Hilbert (RH)  problem.
In \cite{DeiftZhou}, Deift and Zhou introduced the powerful nonlinear steepest descent method for the analysis of RH problems. The method employs a sequence of transformations to obtain the ``small-norm'' RH problem (for a standard textbook in this field, see \cite{Deift2000}), from which the asymptotics of the Hankel determinants are then derived.

\subsubsection{The Method of Refinement} \label{subsubsec: refinement}

By Taylor expansion and the definition of $h_N$ in \eqref{definition of hN}, one has
\begin{equation}\label{Taylor expansion of hN}
h_N(\lambda_j) = \sqrt{2}\pi N F'(\kappa_j)(\kappa_j - \lambda_j) + \frac{\sqrt{2}\pi N}{2} F''(\zeta_j)(\kappa_j - \lambda_j)^2,  
\end{equation}
where $\zeta_j$ is some value between $\kappa_j$ and $\lambda_j$. Then one can easily see that for $j\asymp N$ Theorem \ref{maximum of hN} implies the same estimation of $F'(\kappa_j)|\lambda_j - \kappa_j|$ as in Theorem \ref{main theorem}. But when $j$ or $1-j$ is of order $o(N)$, Theorem \ref{maximum of hN} cannot provide our final results directly.

The idea of refinement comes from the following observation: as $N$ becomes sufficiently large, the jumps in $h_N(x)$ are of order $1/N$; hence, $h_N(x)$ will look ``more like'' a continuous function. In addition, we have $h_N(-1) = 0$ and $\lim_{N\to\infty} \mathbb{P}(h_N(1) = 0) = 1$, which implies that the order of $h_N(x)$ is not likely to be larger than $O(\log N)$ as it goes to the edge regime. Therefore, we use the method of iteration, combined with the argument of contradiction, to refine the upper bound of $h_N(x)$ near $\pm 1$. To be specific, in Lemmas \ref{hN near edge first lemma}, \ref{hN near edge second lemma}, and \ref{hN near soft edge first lemma}, we can see $h_N(x)$ is of order $o(\log N)$ as $x$ approaches $-1$ or $1$. Then we use this new estimation to ``push'' the estimation of $F'(\kappa_j)|\lambda_j - \kappa_j|$ from the bulk regime to the edge regime and obtain our final result.

\subsection{Organization of the Paper}

The rest of the paper is arranged as follows. The RH problem and the differential identities for the Hankel determinants are introduced in Section 2. In Section 3, using the Deift–Zhou nonlinear steepest descent method, the original RH problem is transformed into small-norm ones. Based on the weight of the Hankel determinants, three distinct cases are considered: (i) with two jumps at $x_1$, $x_2$, separated from each other and from $\pm 1$; (ii) with two jumps at $x_1$, $x_2$, close together but away from $\pm 1$; (iii) with a single jump at $x$, near $-1$ or $1$. 
In Section 4, the asymptotics of the Hankel determinants are derived via steepest descent analysis for each of these three cases.
Section 5 examines the maximum of $h_N(x)$, provides eigenvalue rigidity and verifies the main theorem.

\section{The Riemann-Hilbert Problems for the Hankel Determinants}

We study the following RH problem for $Y$.
\begin{rhp}\label{rhp for Y}
 \hfill
 \begin{itemize}
 \item[(a)] $Y: \mathbb{C}\setminus [-1, \infty) \to \mathbb{C}^{2\times 2}$ is analytic.

 \item[(b)] $Y_+(x) = Y_-(x)\left(\begin{array}{cc} 1 & t(x)\\ 0 & 1 \end{array}\right)$ for $x\in [-1, \infty)\setminus\{x_1, x_2\}$, where
  \begin{equation}\label{RHP-Y-weight}
  t(x) = e^{\sqrt{2}\pi\gamma_1 1_{(-1, x_1]}(x)+\sqrt{2}\pi\gamma_2 1_{(-1, x_2]}} w_L(x),
  \end{equation}
  where $x\in [-1, \infty)$,
    \begin{equation}\label{weight}
    w_L(x) = e^{-NV(x)}(x+1)^{\alpha},
    \end{equation}
    and $\alpha> -1$.

 \item[(c)] As $z\to\infty$, $Y(z)$ has the following asymptotics
  \[
  Y(z) = \left( I+O\left(\frac{1}{z}\right)\right) \left(\begin{array}{cc} z^N & 0\\ 0 & z^{-N} \end{array}\right)
  \]

 \item[(d)] As $z\to x_j, j=1, 2$, we have
  \[
  Y(z) = O(\log|z-x_j|).
  \]

  \item[(e)] As $z\to -1$, we have
  \begin{equation}
  Y(z) = \left\{
  \begin{array}{ll}
  O\left(\begin{array}{cc} 1 & |z+1|^\alpha \\ 1 & |z+1|^\alpha \end{array}\right) & \text{if } \alpha<0;\\
  O\left(\begin{array}{cc} 1 & \log|z+1| \\ 1 & \log|z+1| \end{array}\right) & \text{if } \alpha=0,\\
  O\left(\begin{array}{cc} 1 & 1 \\ 1 & 1 \end{array}\right) & \text{if } \alpha>0.
  \end{array}
  \right.
  \end{equation}
 \end{itemize}
\end{rhp}

The solution to this RH problem is (e.g. see \cite{ABJ2004})
\begin{eqnarray} \label{Y}
Y(z) = Y_N(z; x_1, x_2; \gamma_1, \gamma_2;  w_L)=
\left(
\begin{array}{cc}
\frac{1}{ \varkappa_N}p_N(z) & \frac{1}{\varkappa_N} C_t p_N(z) \\
-2\pi i \varkappa_{N-1} p_{N-1}(z) & -2\pi i \varkappa_{N-1} C_t p_{N-1}(z)\\
\end{array}
\right),
\end{eqnarray}
where $p_n(x) =  \varkappa_n x^n + \cdots$ is the orthonormal polynomial with respect to the weight function  $t(x)$. This is analytic in $\mathbb{C}\setminus [-1, \infty)$ and is the unique solution to the Riemann-Hilbert problem for $Y$. Here $C_t$ is the Cauchy transform for $z\in \mathbb{C}\setminus [-1, \infty)$ given by
\begin{eqnarray}
C_t g(z) = \frac{1}{2\pi i}\int_{-1}^\infty g(x)t(x) \frac{dx}{x-z}.
\end{eqnarray}

We will then give two differential identities that will be used for the asymptotics of the Hankel determinants. It is similar to \cite[Prop 5.1]{CFL2021}. We state it as the following proposition.

\begin{proposition}\label{differential identities of Hankel determinants}
Given $\gamma_1, \gamma_2 \in \mathbb{R}$ and $w\equiv 0$, the following result holds for $-1 < x_1 < x_2 < 1$.
\begin{equation}\label{di1}
\frac{d}{dy} \log D_N(x_1, x_2 = x_1+y; \gamma_1, \gamma_2; 0) = - (1+x_2)^\alpha e^{-NV(x_2)} \frac{1-e^{\sqrt{2}\pi\gamma_2}}{2\pi i} \left(Y^{-1}(x_2) Y'(x_2)\right)_{2 1},
\end{equation}
where by $\left(Y^{-1}(x_2) Y'(x_2)\right)_{2 1}$ we mean the limit taken as $z\to x_2$ from both $\mathbb{C}^+$ and $\mathbb{C}^-$. If $x_1 = x_2 = x \in (-1, 1)$ and let $\gamma = \gamma_1+\gamma_2$, we have
\begin{equation}\label{di2}
\frac{d}{dx} \log D_N(x; \gamma; 0) = -(1+x)^\alpha e^{-NV(x)} \frac{1 - e^{\sqrt{2}\pi\gamma}}{2\pi i} \left(Y^{-1}(x) Y'(x)\right)_{2, 1}.
\end{equation}
\end{proposition}

\begin{proof}
The relevant proof is referred to \cite[Sec. 3]{Charlier-IMRN2019}; similar differential identities also appear in \cite{CG2021, CFL2021}.
\end{proof}

Next, a series of transformations is carried out on RH problem \ref{rhp for Y}, with the aim of converting it into the so-called ``small norm'' RH problem so that the asymptotics can be studied.

\section{Steepest Descent Analysis for the RH Problem} \label{Steepest Descent Analysis for the RH Problem}

\subsection{Normalization: $Y\mapsto T$}

Now, the transformation $Y\mapsto T$ is made to normalize the behavior of $Y$  at infinity. Define
\begin{equation}
\rho(x) = \psi_V(x) \sqrt{\frac{1-x}{1+x}},
\end{equation}
where $\psi_V(x)$ is given in \eqref{LUE eq measure}.
Then let
\begin{equation}
T(z) = e^{\frac{Nl}{2}\sigma_3} Y(z) e^{-Ng(z)\sigma_3} e^{-\frac{Nl}{2}\sigma_3}
\end{equation}
for $z\in \mathbb{C}\setminus [-1, \infty)$, where $\sigma_3$ is the third Pauli matrix, defined as
\begin{equation*}
\sigma_3 = \left(\begin{array}{cc} 1 & 0 \\ 0 & -1 \end{array}\right).
\end{equation*}
And
\begin{equation}
g(z) = \int_{-1}^1 \log (z-s) \rho(s) ds, \text{ for } z\in \mathbb{C} \setminus (-\infty, 1],
\end{equation}
where we choose the principal branch for the logarithm function. Hence $g(z)$ is analytic in $\mathbb{C}\setminus (-\infty, 1]$ and has the following properties.
\begin{eqnarray}
g_+(x) + g_-(x) &=& 2 \int_{-1}^1 \log |x-s|\rho(s)ds,  x\in \mathbb{R}; \\
g_+(x) - g_-(x) &=& 2\pi i,  x\in (-\infty, -1); \\
g_+(x) - g_-(x) &=& 2\pi i \int_x^1 \rho(s)ds,  x\in [-1, 1].
\end{eqnarray}

We also define
\begin{equation}
\tilde{\rho}(z) = -i \psi_V(z) \sqrt{\frac{z-1}{z+1}}
\end{equation}
for $z\in U_V \setminus [-1,1]$, where $U_V$ is a neighbourhood of $[-1, \infty)$. Then define
\begin{equation} \label{xi(z)}
\xi(z) = -\pi i \int_{1}^z \tilde{\rho}(s)ds, \quad z\in U_V \setminus (-\infty, 1),
\end{equation}
 where the path of integration lies in $U_V\setminus (-\infty, 1)$. It is direct to check that $\xi_+(x) + \xi_{-}(x) = 0$ for $x\in (-1, 1)$, and then 
\begin{equation}
2\xi_{\pm}(x) = g_{\pm}(x) - g_{\mp}(x) = 2g_{\pm}(x) - V(x) + l.
\end{equation}
By analytic continuation, one has
\begin{equation}
\xi(z) = g(z) + \frac{l}{2} - \frac{V(z)}{2}, \quad z\in U_V \setminus (-\infty, 1).
\end{equation}
Then we have the following RH problem for $T$.

\begin{rhp}
 \hfill
 \begin{itemize}
 \item[(a)] $T: \mathbb{C}\setminus [-1, \infty) \to \mathbb{C}^{2\times 2}$ is analytic.

 \item[(b)] $T$ satisfies the following jump relations
 \begin{eqnarray}
 T_+(x) = T_-(x)
 \left\{
 \begin{array}{ll}
 \left(\begin{array}{cc} e^{-2N\xi_+(x)} & e^{\sqrt{2}\pi\gamma_1+\sqrt{2}\pi\gamma_2}(x+1)^{\alpha} \\ 0 & e^{2N\xi_+(x)} \end{array} \right), & \text{ for } -1<x<x_1;\\
 \left(\begin{array}{cc} e^{-2N\xi_+(x)} & e^{\sqrt{2}\pi\gamma_2}(x+1)^{\alpha} \\ 0 & e^{2N\xi_+(x)} \end{array} \right), & \text{ for } x_1<x<x_2;  \\
 \left(\begin{array}{cc} e^{-2N\xi_+(x)} & (x+1)^{\alpha} \\ 0 & e^{2N\xi_+(x)} \end{array} \right), & \text{ for } x_2<x<1;\\
 \left(\begin{array}{cc} 1 & (x+1)^{\alpha}e^{2N\xi(x)} \\ 0 & 1 \end{array} \right), & \text{ for } x>1;
 \end{array}
 \right. \label{T-jump-2}
 \end{eqnarray}

 \item[(c)]  $T(z)$ exhibits the following behavior at infinity
  \[
  T(z) = I+O\left(\frac{1}{z}\right), \ z\to\infty
  \]

 \item[(d)] When $z\to x_j, j=1, 2$,
  \[
  T(z) = O(\log|z-x_j|).
  \]

 \item[(e)] $T(z)$ displays the following behavior at  $ -1$, we have
  \begin{equation}
  T(z) = \left\{
  \begin{array}{ll}
  O\left(\begin{array}{cc} 1 & |z+1|^\alpha \\ 1 & |z+1|^\alpha \end{array}\right) & \text{if } \alpha<0;\\
  O\left(\begin{array}{cc} 1 & \log|z+1| \\ 1 & \log|z+1| \end{array}\right) & \text{if } \alpha=0,\\
  O\left(\begin{array}{cc} 1 & 1 \\ 1 & 1 \end{array}\right) & \text{if } \alpha>0.
  \end{array}
  \right.
  \end{equation}
 \item[(f)]  As $z\to 1$, $T(z)$ is bounded.

 \end{itemize}
\end{rhp}

\subsection{Contour Deformation: $T\mapsto S$}

The strong oscillation of the diagonal entries $e^{\pm2N\xi_+(x)}$ when $N\to \infty$,  resulting from the condition $\Re \xi_{\pm}(x) = 0$ on $x\in (-1, 1)$, motivates a contour deformation and a second transformation $T\to S$ for the removal of the oscillatory terms. The procedure relies on the matrix factorization
\begin{eqnarray}
& & \left(\begin{array}{cc} e^{-2N\xi_+(x)} & e^{\sqrt{2}\pi\gamma}(x+1)^{\alpha} \\ 0 & e^{2N\xi_+(x)} \end{array} \right)\nonumber\\
&=& \left(\begin{array}{cc} 1 & 0 \\ e^{-\sqrt {2}\pi\gamma}(x+1)^{-\alpha} e^{2N\xi_{+}(x)} & 1\end{array}\right) \left(\begin{array}{cc} 0 & e^{\sqrt {2}\pi\gamma}(x+1)^{\alpha} \\ -e^{-\sqrt {2}\pi\gamma}(x+1)^{-\alpha} & 0\end{array}\right) \nonumber\\
& & \times\left(\begin{array}{cc} 1 & 0 \\ e^{-\sqrt {2}\pi\gamma}(x+1)^{-\alpha} e^{-2N\xi_+(x)} & 1 \end{array}\right) \nonumber\\
&=:& J_1(x; \gamma)J_2(x; \gamma)J_3(x; \gamma),
\end{eqnarray}
where the parameter $\gamma$ assumes the value $\gamma_1+\gamma_2$ on the interval $ x<x_1$,  $\gamma_2$ on $x_1< x<x_2$, and $0$ on $ x>x_2$.

The properties of the jump points in the original weight function $t(x)$ in \eqref{RHP-Y-weight} give rise to the following four cases in the subsequent analysis:
\begin{itemize}
\item[(I)] with two jumps at $x_1$, $x_2$, separated from each other and from $\pm 1$ (see Figure \ref{fig:contours S});

\item[(II)] with two jumps at $x_1$, $x_2$, close together but away from $\pm 1$ (see Figure  \ref{fig:contours 2});

\item[(III)] with a single jump at $x$, near $-1$ (see Figure \ref{fig:contours 3});

\item[(IV)] with a single jump at $x$, near $1$ (see Figure \ref{fig:contours 4}).
\end{itemize}

\tikzset{every picture/.style={line width=0.75pt}} 

\begin{figure}[htbp]
\centering

\tikzset{every picture/.style={line width=0.75pt}} 

\begin{tikzpicture}[x=0.75pt,y=0.75pt,yscale=-1,xscale=1]

\draw    (100,113) -- (176.58,113.25) -- (253.17,113.5) -- (329.75,113.75) -- (406.33,114) ;
\draw [shift={(143.29,113.14)}, rotate = 180.19] [fill={rgb, 255:red, 0; green, 0; blue, 0 }  ][line width=0.08]  [draw opacity=0] (8.93,-4.29) -- (0,0) -- (8.93,4.29) -- cycle    ;
\draw [shift={(219.87,113.39)}, rotate = 180.19] [fill={rgb, 255:red, 0; green, 0; blue, 0 }  ][line width=0.08]  [draw opacity=0] (8.93,-4.29) -- (0,0) -- (8.93,4.29) -- cycle    ;
\draw [shift={(296.46,113.64)}, rotate = 180.19] [fill={rgb, 255:red, 0; green, 0; blue, 0 }  ][line width=0.08]  [draw opacity=0] (8.93,-4.29) -- (0,0) -- (8.93,4.29) -- cycle    ;
\draw [shift={(373.04,113.89)}, rotate = 180.19] [fill={rgb, 255:red, 0; green, 0; blue, 0 }  ][line width=0.08]  [draw opacity=0] (8.93,-4.29) -- (0,0) -- (8.93,4.29) -- cycle    ;
\draw    (100,113) .. controls (129.33,81) and (152.33,82) .. (176.58,113.25) ;
\draw [shift={(143.91,89.69)}, rotate = 182.17] [fill={rgb, 255:red, 0; green, 0; blue, 0 }  ][line width=0.08]  [draw opacity=0] (8.93,-4.29) -- (0,0) -- (8.93,4.29) -- cycle    ;
\draw    (176.58,113.25) .. controls (205.92,81.25) and (228.92,82.25) .. (253.17,113.5) ;
\draw [shift={(220.5,89.94)}, rotate = 182.17] [fill={rgb, 255:red, 0; green, 0; blue, 0 }  ][line width=0.08]  [draw opacity=0] (8.93,-4.29) -- (0,0) -- (8.93,4.29) -- cycle    ;
\draw    (253.17,113.5) .. controls (282.5,81.5) and (305.5,82.5) .. (329.75,113.75) ;
\draw [shift={(297.08,90.19)}, rotate = 182.17] [fill={rgb, 255:red, 0; green, 0; blue, 0 }  ][line width=0.08]  [draw opacity=0] (8.93,-4.29) -- (0,0) -- (8.93,4.29) -- cycle    ;
\draw    (100,113) .. controls (126.33,142) and (153.33,142) .. (176.58,113.25) ;
\draw [shift={(143.71,134.53)}, rotate = 178.22] [fill={rgb, 255:red, 0; green, 0; blue, 0 }  ][line width=0.08]  [draw opacity=0] (8.93,-4.29) -- (0,0) -- (8.93,4.29) -- cycle    ;
\draw    (176.58,113.25) .. controls (202.92,142.25) and (229.92,142.25) .. (253.17,113.5) ;
\draw [shift={(220.29,134.78)}, rotate = 178.22] [fill={rgb, 255:red, 0; green, 0; blue, 0 }  ][line width=0.08]  [draw opacity=0] (8.93,-4.29) -- (0,0) -- (8.93,4.29) -- cycle    ;
\draw    (253.17,113.5) .. controls (279.5,142.5) and (306.5,142.5) .. (329.75,113.75) ;
\draw [shift={(296.87,135.03)}, rotate = 178.22] [fill={rgb, 255:red, 0; green, 0; blue, 0 }  ][line width=0.08]  [draw opacity=0] (8.93,-4.29) -- (0,0) -- (8.93,4.29) -- cycle    ;

\draw (81,114) node [anchor=north west][inner sep=0.75pt]   [align=left] {$\displaystyle -1$};
\draw (168,116) node [anchor=north west][inner sep=0.75pt]   [align=left] {$\displaystyle x_{1}$};
\draw (244,122) node [anchor=north west][inner sep=0.75pt]   [align=left] {$\displaystyle x_{2}$};
\draw (329.75,113.75) node [anchor=north west][inner sep=0.75pt]   [align=left] {$\displaystyle 1$};

\end{tikzpicture}
\caption{The jump contours for the RH problem for $S$}\label{fig:contours S}
\end{figure}
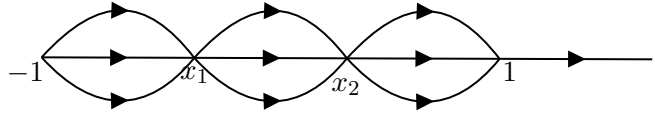

\tikzset{every picture/.style={line width=0.75pt}} 

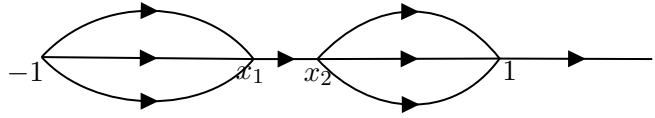
\begin{figure}[htbp]
\centering

\tikzset{every picture/.style={line width=0.75pt}} 

\begin{tikzpicture}[x=0.75pt,y=0.75pt,yscale=-1,xscale=1]

\draw    (100,113) -- (206.33,114) -- (238.33,114) -- (329.75,113.75) -- (406.33,114) ;
\draw [shift={(158.17,113.55)}, rotate = 180.54] [fill={rgb, 255:red, 0; green, 0; blue, 0 }  ][line width=0.08]  [draw opacity=0] (8.93,-4.29) -- (0,0) -- (8.93,4.29) -- cycle    ;
\draw [shift={(227.33,114)}, rotate = 180] [fill={rgb, 255:red, 0; green, 0; blue, 0 }  ][line width=0.08]  [draw opacity=0] (8.93,-4.29) -- (0,0) -- (8.93,4.29) -- cycle    ;
\draw [shift={(289.04,113.86)}, rotate = 179.84] [fill={rgb, 255:red, 0; green, 0; blue, 0 }  ][line width=0.08]  [draw opacity=0] (8.93,-4.29) -- (0,0) -- (8.93,4.29) -- cycle    ;
\draw [shift={(373.04,113.89)}, rotate = 180.19] [fill={rgb, 255:red, 0; green, 0; blue, 0 }  ][line width=0.08]  [draw opacity=0] (8.93,-4.29) -- (0,0) -- (8.93,4.29) -- cycle    ;
\draw    (100,113) .. controls (129.33,81) and (182.08,82.75) .. (206.33,114) ;
\draw [shift={(158.64,89.93)}, rotate = 181.45] [fill={rgb, 255:red, 0; green, 0; blue, 0 }  ][line width=0.08]  [draw opacity=0] (8.93,-4.29) -- (0,0) -- (8.93,4.29) -- cycle    ;
\draw    (238.33,114) .. controls (267.67,82) and (305.5,82.5) .. (329.75,113.75) ;
\draw [shift={(289.33,90.28)}, rotate = 180.67] [fill={rgb, 255:red, 0; green, 0; blue, 0 }  ][line width=0.08]  [draw opacity=0] (8.93,-4.29) -- (0,0) -- (8.93,4.29) -- cycle    ;
\draw    (100,113) .. controls (126.33,142) and (183.08,142.75) .. (206.33,114) ;
\draw [shift={(158.6,135.1)}, rotate = 179.8] [fill={rgb, 255:red, 0; green, 0; blue, 0 }  ][line width=0.08]  [draw opacity=0] (8.93,-4.29) -- (0,0) -- (8.93,4.29) -- cycle    ;
\draw    (238.33,114) .. controls (264.67,143) and (306.33,146) .. (329.75,113.75) ;
\draw [shift={(289.62,136.75)}, rotate = 179.59] [fill={rgb, 255:red, 0; green, 0; blue, 0 }  ][line width=0.08]  [draw opacity=0] (8.93,-4.29) -- (0,0) -- (8.93,4.29) -- cycle    ;

\draw (81,114) node [anchor=north west][inner sep=0.75pt]   [align=left] {$\displaystyle -1$};
\draw (196,116) node [anchor=north west][inner sep=0.75pt]   [align=left] {$\displaystyle x_{1}$};
\draw (230,117) node [anchor=north west][inner sep=0.75pt]   [align=left] {$\displaystyle x_{2}$};
\draw (329.75,113.75) node [anchor=north west][inner sep=0.75pt]   [align=left] {$\displaystyle 1$};

\end{tikzpicture}

\caption{The jump contours for the RH problem for merging case}\label{fig:contours 2}
\end{figure}

\tikzset{every picture/.style={line width=0.75pt}} 

\begin{figure}[htbp]
\centering

\tikzset{every picture/.style={line width=0.75pt}} 

\begin{tikzpicture}[x=0.75pt,y=0.75pt,yscale=-1,xscale=1]

\draw    (100,109) -- (142.33,110) -- (260.22,111) -- (340.33,112) ;
\draw [shift={(126.17,109.62)}, rotate = 181.35] [fill={rgb, 255:red, 0; green, 0; blue, 0 }  ][line width=0.08]  [draw opacity=0] (8.93,-4.29) -- (0,0) -- (8.93,4.29) -- cycle    ;
\draw [shift={(206.28,110.54)}, rotate = 180.49] [fill={rgb, 255:red, 0; green, 0; blue, 0 }  ][line width=0.08]  [draw opacity=0] (8.93,-4.29) -- (0,0) -- (8.93,4.29) -- cycle    ;
\draw [shift={(305.28,111.56)}, rotate = 180.72] [fill={rgb, 255:red, 0; green, 0; blue, 0 }  ][line width=0.08]  [draw opacity=0] (8.93,-4.29) -- (0,0) -- (8.93,4.29) -- cycle    ;
\draw    (100,109) .. controls (110.33,87) and (131.33,86) .. (142.33,110) ;
\draw [shift={(126.6,93.25)}, rotate = 189.58] [fill={rgb, 255:red, 0; green, 0; blue, 0 }  ][line width=0.08]  [draw opacity=0] (8.93,-4.29) -- (0,0) -- (8.93,4.29) -- cycle    ;
\draw    (101.28,110.5) .. controls (107.33,133) and (133.33,136) .. (142.33,110) ;
\draw [shift={(126.74,127.6)}, rotate = 171.67] [fill={rgb, 255:red, 0; green, 0; blue, 0 }  ][line width=0.08]  [draw opacity=0] (8.93,-4.29) -- (0,0) -- (8.93,4.29) -- cycle    ;
\draw    (142.33,110) .. controls (174.33,77) and (221.33,77) .. (260.22,111) ;
\draw [shift={(206.15,85.78)}, rotate = 183.82] [fill={rgb, 255:red, 0; green, 0; blue, 0 }  ][line width=0.08]  [draw opacity=0] (8.93,-4.29) -- (0,0) -- (8.93,4.29) -- cycle    ;
\draw    (142.33,110) .. controls (175.33,143) and (222.33,145) .. (260.22,111) ;
\draw [shift={(206.26,135.44)}, rotate = 178.25] [fill={rgb, 255:red, 0; green, 0; blue, 0 }  ][line width=0.08]  [draw opacity=0] (8.93,-4.29) -- (0,0) -- (8.93,4.29) -- cycle    ;

\draw (79,114) node [anchor=north west][inner sep=0.75pt]   [align=left] {$\displaystyle -1$};
\draw (137,118) node [anchor=north west][inner sep=0.75pt]   [align=left] {$\displaystyle x$};
\draw (253,118) node [anchor=north west][inner sep=0.75pt]   [align=left] {$\displaystyle 1$};

\end{tikzpicture}
\caption{The jump contours for the RH problem for edge regime near $-1$}
\label{fig:contours 3}
\end{figure}
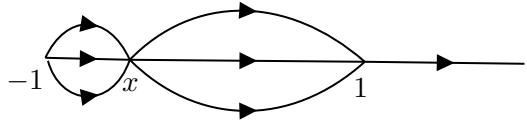

\tikzset{every picture/.style={line width=0.75pt}} 

\begin{figure}[htbp]
\centering

\tikzset{every picture/.style={line width=0.75pt}} 

\begin{tikzpicture}[x=0.75pt,y=0.75pt,yscale=-1,xscale=1]

\draw    (100,109) -- (215.33,110) -- (260.22,111) -- (340.33,112) ;
\draw [shift={(162.67,109.54)}, rotate = 180.5] [fill={rgb, 255:red, 0; green, 0; blue, 0 }  ][line width=0.08]  [draw opacity=0] (8.93,-4.29) -- (0,0) -- (8.93,4.29) -- cycle    ;
\draw [shift={(242.78,110.61)}, rotate = 181.28] [fill={rgb, 255:red, 0; green, 0; blue, 0 }  ][line width=0.08]  [draw opacity=0] (8.93,-4.29) -- (0,0) -- (8.93,4.29) -- cycle    ;
\draw [shift={(305.28,111.56)}, rotate = 180.72] [fill={rgb, 255:red, 0; green, 0; blue, 0 }  ][line width=0.08]  [draw opacity=0] (8.93,-4.29) -- (0,0) -- (8.93,4.29) -- cycle    ;
\draw    (100,109) .. controls (139.33,74) and (181.33,77) .. (215.33,110) ;
\draw [shift={(163.2,84.18)}, rotate = 181.95] [fill={rgb, 255:red, 0; green, 0; blue, 0 }  ][line width=0.08]  [draw opacity=0] (8.93,-4.29) -- (0,0) -- (8.93,4.29) -- cycle    ;
\draw    (100,109) .. controls (134.33,141) and (177.33,146) .. (215.33,110) ;
\draw [shift={(162.84,134.94)}, rotate = 179.2] [fill={rgb, 255:red, 0; green, 0; blue, 0 }  ][line width=0.08]  [draw opacity=0] (8.93,-4.29) -- (0,0) -- (8.93,4.29) -- cycle    ;
\draw    (215.33,110) .. controls (230.33,85) and (247.33,85) .. (260.22,111) ;
\draw [shift={(243.4,92.3)}, rotate = 188.49] [fill={rgb, 255:red, 0; green, 0; blue, 0 }  ][line width=0.08]  [draw opacity=0] (8.93,-4.29) -- (0,0) -- (8.93,4.29) -- cycle    ;
\draw    (215.33,110) .. controls (230.33,133) and (247.33,133) .. (260.22,111) ;
\draw [shift={(242.72,126.9)}, rotate = 176.34] [fill={rgb, 255:red, 0; green, 0; blue, 0 }  ][line width=0.08]  [draw opacity=0] (8.93,-4.29) -- (0,0) -- (8.93,4.29) -- cycle    ;

\draw (79,114) node [anchor=north west][inner sep=0.75pt]   [align=left] {$\displaystyle -1$};
\draw (210.33,113) node [anchor=north west][inner sep=0.75pt]   [align=left] {$\displaystyle x$};
\draw (253,118) node [anchor=north west][inner sep=0.75pt]   [align=left] {$\displaystyle 1$};

\end{tikzpicture}
\caption{The jump contours for the RH problem for edge regime near $1$}
\label{fig:contours 4}
\end{figure}
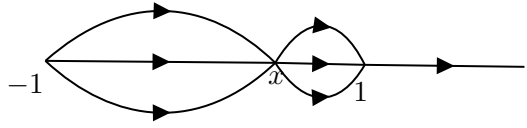

We will discuss the second transformation $T\mapsto S$ in Case (I) in detail. The weight $w_L(x)$ defined in \eqref{weight} can be analytically extended to $U_V \setminus (-\infty, -1]$, where $U_V$ is a neighboorhood of $[-1, \infty)$, and we still denote it by $w_L(z)$, which is given as 
\begin{equation}\label{eq: Laguerre weight}
w_L(z) = e^{-NV(z)}(z+1)^{\alpha}.
\end{equation}

Then we open the lens and define $S$ by
\begin{equation} \label{eq:T-S-map}
S(z) = \left\{\begin{array}{ll}T(z)J_3(z; \gamma)^{-1}, & z\in\Omega,\\ T(z)J_1(z;\gamma), & z\in\bar{\Omega},\\ T(z), & \text{elsewhere},
\end{array}\right.
\end{equation}
where $\Omega$ denotes the upper lens, and $\bar \Omega$ denotes the lower lens.

Now the RH problem for $S$ becomes the following.

\begin{rhp}
 \hfill
 \begin{itemize}
 \item[(a)] $S: \mathbb{C}\setminus ([-1, \infty)\cup \Gamma_1\cup\overline{\Gamma}_1\cup \Gamma_2\cup\overline{\Gamma}_2\cup\Gamma_3\cup\overline{\Gamma}_3) \to \mathbb{C}^{2\times 2}$ is analytic. Here $\Gamma_i$-s denote the contours of the upper lens and $\bar \Gamma_i$-s denote the contours of the lower lens.

 \item[(b)] $S$ satisfies the following jump relations
 \begin{equation}
 S_+(z) = S_-(z) J_S(z), \qquad z \in \Sigma_S
\end{equation} 
with
 \begin{eqnarray} \label{jumps for S}
 J_S(z) = \begin{cases}
J_1(z; \gamma), & z\in \overline{\Gamma}_j, j=1,2,3,\\
J_2(x; \gamma), & -1<x<1,\\
J_3(z; \gamma), & z\in \Gamma_j, j=1,2,3,\\
\left(\begin{array}{cc} 1 & (x+1)^{\alpha}e^{2N\xi(x)} \\ 0 & 1 \end{array} \right) & x>1.
 \end{cases}
 \end{eqnarray}

 \item[(c)] As $z\to\infty$, $S(z)$ has the following asymptotics
  \[
  S(z) = I+O\left(\frac{1}{z}\right)
  \]

 \item[(d)] As $z\to x_j, j=1, 2$, we have
  \[
  S(z) = O(\log|z-x_j|),
  \]
  where $z$ may approach $x_j$ from either inside or outside the lens.

 \item[(e)]
 For $\alpha<0$, the matrix function $S(z)$ has the following behavior as $z\to -1$:
 \begin{equation*}
 S(z) =
  O\left(\begin{array}{cc} 1 & |z+1|^\alpha \\ 1 & |z+1|^\alpha \end{array}\right), \text{ as }z\to -1, z\in \mathbb{C}\setminus\Sigma.
 \end{equation*}
 For $\alpha=0$, $S(z)$ has the following behavior as $z\to -1$:
 \begin{equation*}
 S(z) =
  O\left(\begin{array}{cc} \log|z+1| & \log|z+1| \\ \log|z+1| & \log|z+1| \end{array}\right), \text{ as }z\to -1, z\in \mathbb{C}\setminus\Sigma.
 \end{equation*}
 For $\alpha>0$, $S(z)$ has the following behavior as $z\to -1$:
 \begin{equation*}
 S(z) = \left\{\begin{array}{ll} O\left(\begin{array}{cc} 1 & 1\\ 1 & 1 \end{array}\right), & \text{as $z\to -1$ outside the lens},\\
 O\left(\begin{array}{cc} |z+1|^{-\alpha} & 1\\ |z+1|^{-\alpha} & 1 \end{array}\right), & \text{as $z\to -1$ inside the lens},
 \end{array}
 \right.
 \end{equation*}
 \item[(f)] As $z\to 1$, $S(z)$ is bounded.
 \end{itemize}
\end{rhp}

In Case (II), the two jump points $x_1, x_2$ are sufficiently close , so the lens over $[x_1,x_2]$ is left unopened (see Figure \ref{fig:contours 2}). Hence, $T(x)$ and $S(x)$ share the same jump matrix on this interval.
In Cases (III) and (IV), where the single jump point $x$ is close to the endpoints $-1$ and $1$, respectively. Accordingly, the transformation $T \mapsto S$ in \eqref{eq:T-S-map} is defined as shown in Figure \ref{fig:contours 3} and \ref{fig:contours 4}, respectively.


\subsection{The Global Parametrix}

For $z$ bounded away from $[-1, 1]$, one can see that all the jump matrices tend to the identity as $N \to \infty$. We approximate $S$ by the global parametrix $P^\infty(z)$. Note that in this case, we only have jumps on the interval $[-1, 1]$. Hence, we consider the following RH problem for $P^\infty(z)$.

\begin{rhp}\label{rhp for Pinfty}
\hfill
\begin{itemize}
\item[(a)] $P^\infty: \mathbb{C}\setminus [-1, 1] \to \mathbb{C}^{2\times 2}$ is analytic.

\item[(b)] $P^\infty$ satisfies the following jump relations
 \begin{eqnarray*}
 \begin{array}{lr}
 P^\infty_+(x) = P^\infty_-(x) J_2(x; \gamma_1+\gamma_2), & -1<x<x_1,\\
 P^\infty_+(x) = P^\infty_-(x) J_2(x; \gamma_2), & x_1<x<x_2,\\
 P^\infty_+(x) = P^\infty_-(x) J_2(x; 0), & x_2 < x <1.
 \end{array}
 \end{eqnarray*}

 \item[(c)] As $z\to\infty$, we have
 \begin{equation}
 P^{\infty}(z) = I + O(z^{-1}).
 \end{equation}

 \item[(d)] As $z\to x_j$ for $j=1, 2$, we have
 \begin{equation}
 P^{\infty}(z) = O(\log |z-x_j|).
 \end{equation}

 \item[(e)] As $z\to -1$, we have
 \begin{equation}
 P^{\infty}(z) = O\left(
                    \begin{array}{cc}
                      |z + 1|^{-1/4} & |z + 1|^{-1/4} \\
                      |z + 1|^{-1/4} & |z + 1|^{-1/4} \\
                    \end{array}
                  \right) (z+1)^{-\frac{\alpha}{2}\sigma_3}.
 \end{equation}

 \item[(f)] As $z\to 1$, we have 
 \begin{equation}
 P^{\infty}(z) = O((z-1)^{-1/4})
 \end{equation}

\end{itemize}
\end{rhp}

 Note that here we impose slightly different conditions as $z\to -1$, in order for $P^{\infty}(z)$ to be solvable. Our case is a special case considered in \cite[Sec. 5.4]{CG2021}; hence, the solution is explicitly given by
\begin{equation}\label{Pinfty}
P^\infty (z) = D_{\infty}^{\sigma_3} Q(z) D(z)^{-\sigma_3},
\end{equation}
where
\begin{eqnarray}\label{Q}
Q(z) = \left(
         \begin{array}{cc}
           \frac{1}{2}(a(z) + a(z)^{-1}) & -\frac{1}{2i}(a(z) - a(z)^{-1}) \\
           \frac{1}{2i}(a(z) - a(z)^{-1}) & \frac{1}{2}(a(z) + a(z)^{-1}) \\
         \end{array}
       \right)
\end{eqnarray}
and
\begin{equation}
a(z) = \left(\frac{z+1}{z-1}\right)^{1/4}.
\end{equation}
And $D(z)$ takes the form of
\begin{equation}\label{D}
D(z) = D_{\alpha}(z) D_w(z) D_{\gamma}(z),
\end{equation}
where
\begin{eqnarray}
D_{\alpha}(z) &=& (z+\sqrt{z^2-1})^{-\alpha / 2} (z+1)^{\alpha / 2}; \\
D_{t}(z) &=&  \exp\left(\frac{(z^2-1)^{1/2}}{2\pi} \int_{-1}^1 \frac{t(x)}{\sqrt{1-x^2}} \frac{dx}{z-x}\right),\\
D_\gamma (z) &=& \exp\left(\frac{(z^2-1)^{1/2}}{\sqrt{2}}\sum_{j=1}^2 \gamma_j \int_{-1}^{x_j} \frac{1}{\sqrt{1-x^2}} \frac{dx}{z-x}\right), \label{Dgamma}
\end{eqnarray}
and $\displaystyle D_{\infty} = \lim_{z\to\infty} D(z)$ is a constant. This is the global parametrix of the RH problem for $S(z)$.

The relevant steepest descent analysis for case (I) has been presented in \cite{CG2021}. Building upon this, we provide the asymptotic expression for the Hankel determinant $D_N(x_1, x_2; \gamma_1, \gamma_2; w)$ in Proposition \ref{Asymptotics in the Separated Regime}.


We now proceed to a detailed analysis of cases (II), (III), and (IV) . In case (II), as previously mentioned, the lens over $[x_1,x_2]$ is left unopened; we retain the global parametrix from \eqref{Pinfty} construct a local parametrix near $x_1$ associated with Painlev\'{e} V, defined in a neighbourhood containing $x_2$. For cases (III) and (IV), with only one jump point $x$, we set $x_1=x_2=x$ and $\gamma_1+\gamma_2 = \gamma$, modifying the global parametrix in \eqref{Pinfty} accordingly, and construct the local parametrix near $x$ using confluent hypergeometric functions.


\subsection{Local Parametrix in the Merging Case}

This section is devoted to the construction of the local parametrix for the merging case (II).
Fix $\delta>0$ and let $B(z_0, \delta):= \{ z \,| \, |z-z_0| < \delta\}$ denote the neighbourhood of $z_0$. First, We consider the RH problem in the neighbourhoods of the endpoints $\pm 1$. For $z$ near $-1$, the local parametrix satisfies the following RH problem:

\begin{rhp} \label{sec3.4-RHP1}
\hfill
\begin{itemize}
\item[(a)] $P_{-1}: B(-1, \delta) \setminus ([-1, \infty) \cup \Gamma_1 \cup \overline{\Gamma}_1) \to \mathbb{C}^{2\times 2}$ is analytic.
\item[(b)] $P_{-1}$ satisfies the same jump relations as $S$ on $B(-1, \delta)\cap ([-1, \infty) \cup \Gamma_1 \cup \overline{\Gamma}_1 \cup \Gamma_2 \cup \overline{\Gamma}_2)$.
\item[(c)] For $z\in \partial B(-1, \delta)$, we have the matching condition
\begin{equation} \label{sec3.3-matching1}
 P_{-1}(z)P^\infty(z)^{-1} = I+O(1/N) \quad \text{ as } N\to\infty.
\end{equation}

\item[(d)] As $z\to -1$, $P_{-1}(z)$ has the same asymptotic behaviours as $S(z)$.
\end{itemize}
\end{rhp}


A solution to the above RH problem is given by the Bessel parametrix; see \cite[Sec. 6]{ABJ2004} for details. The RH problem near $1$ is as follows.

\begin{rhp} \label{sec3.4-RHP-Airy}
\hfill
\begin{itemize}
\item[(a)] $P_{1}: B(1, \delta) \setminus ([-1, \infty) \cup \Gamma_2 \cup \overline{\Gamma}_2)\to \mathbb{C}^{2\times 2}$ is analytic.
\item[(b)] $P_{1}$ satisfies the same jump relations as $S$ on $B(1, \delta)\cap ([-1, \infty) \cup   \Gamma_2 \cup \overline{\Gamma}_2)$.
\item[(c)] For $z\in \partial B(1, \delta)$, we have the matching condition
\begin{equation} \label{sec3.3-matching1-Airy}
P_{1}(z)P^\infty(z)^{-1} = I+O(1/N) \quad \text{ as } N\to\infty.
\end{equation}
\item[(d)] As $z\to 1$, $P_{1}(z)$ is bounded.
\end{itemize}
\end{rhp}

The Airy parametrix provides a solution to the above RH problems; see the explicit construction in \cite{DeiftZhou}.

We will start our discussion on the local parametrix near $x_1$. Let us consider a neighbourhood of $B(x_1, \delta)$ with $\delta > 0$ , which encloses the other jump point $x_2$. We look for a function $P_{x_1}(z)$ satisfying the following RH problem in $B(x_1, \delta)$. 

\begin{rhp} \label{rhp:local-merge}
\hfill
\begin{itemize}
\item[(a)] $P_{x_1}: B(x_1, \delta)\setminus ([-1, 1] \cup \Gamma_1 \cup \overline{\Gamma}_1 \cup \Gamma_2 \cup \overline{\Gamma}_2) \to \mathbb{C}^{2\times 2}$ is analytic.
\item[(b)] $P_{x_1}$ satisfies the same jump condition as  $S$ in $B(x_1, \delta) \cap ([-1, 1]\cup \Gamma_1 \cup \overline{\Gamma}_1 \cup \Gamma_2 \cup \overline{\Gamma}_2)$.
\item[(c)] For $z\in\partial B(x_1, \delta)$, we have the matching condition
\begin{equation} \label{eq: p-p-infty-math}
P_{x_1}(z)P^\infty(z)^{-1} = I + O((\delta N)^{-1})
\end{equation}
as $N\to\infty$, uniformly for all relevant parameters and for $0< y \leq  \delta / 2$ with $y=x_2 - x_1$.
\item[(d)] As $z\to x_j$, $j=1, 2$, we have
\begin{equation}
P_{x_1}(z) = O(\log|z-x_j|).
\end{equation}
\end{itemize}
\end{rhp}


As in \cite[Sec. 7]{CFL2021}, we construct the solution to the above RH problem using the Painlevé V functions. First, in a neighborhood of $z = x_1$, we define the following conformal mapping:
\begin{equation} \label{lambdayz}
\lambda_y(z) = \pm \frac{2N}{s_{N, y}}(\xi(z) - \xi_{\pm} (x_1)), \quad \pm\Im z>0,
\end{equation}
where $\xi(z)$ is defined by \eqref{xi(z)} and $s_{N, y}$ is the quantity measuring the distance between $x_2$ and $x_1$, defined as 
\begin{eqnarray}  \label{sNy} 
s_{N, y} =  2N(\xi_+(x_2) - \xi_+(x_1)) .
\end{eqnarray}
It follows from \eqref{lambdayz} that $\lambda_y(z)$ maps the two jump points $x_1$ and $x_2$ to $0$ and $1$, respectively. For $x \in (-1, 1)\cap B(x_1, \delta)$,
\begin{eqnarray*}
\lambda_{y, +}(x) = \frac{2N}{s_{N, y}}(\xi_{+}(x) - \xi_+(x_1)) = -\frac{2N}{s_{N, y}}(\xi_{-}(x) - \xi_{-}(x_1)) =\lambda_{y, -}(x).
\end{eqnarray*}
Consequently, $\lambda_y(z)$ is analytic in $B(x_1, \delta)$. In addition, \eqref{sNy} gives
\begin{eqnarray}
s_{N, y} = -2Nyi\sqrt{\frac{1-x_1}{1+x_1}} + O(Ny^2), \quad \textrm{as } N\to\infty,
\end{eqnarray}
uniformly for  $0< y \leq \frac{\delta}{2}$. This gives us
\begin{equation}
 \lambda'_y(x_1) = \frac{2N}{s_{N, y}}(\xi_+)'(x_1) = \frac{1}{y} + O(1)
\end{equation}
as $N\to\infty$. In what follows, we present the solution to RH problem \ref{rhp:local-merge}.

\begin{lemma}
Let $\xi(z)$ and $w_L(z)$ be defined in \eqref{xi(z)} and \eqref{eq: Laguerre weight}, and $\widehat{\Phi}_{PV}$ be the Painlev\'e V parametrix given in Appendix \ref{Appendix: Section: PV model RHP}. Then, the solution to RH problem \ref{rhp:local-merge} is given by
\begin{equation}\label{Pmerge}
 P_{x_1}(z) = E_{N, y} (z) \widehat{\Phi}_{PV}(\lambda_y(z); s_{N, y}) (1+z)^{-\frac{\alpha}{2}\sigma_3} e^{-N\xi(z) \sigma_3},
\end{equation}
where $E_{N, y}(z)$ is an analytic function in $B(x_1, \delta)$ given as 
\begin{equation}\label{Emerge}
 E_{N, y}(z) = P^\infty(z) (1+z)^{\frac{\alpha}{2}\sigma_3}\widehat{\Phi}^\infty(\lambda_y(z))^{-1} e^{-N\xi_+(x_1)\sigma_3},
\end{equation}
with
\begin{equation}\label{hatpsiinfty}
\widehat{\Phi}^{\infty}(\lambda) = \widehat{\Phi}^{\infty}(\lambda; s) = \left\{\begin{array}{ll}
e^{-\sqrt{2}\pi\gamma_2\sigma_3} |s|^{\frac{\gamma_1+\gamma_2}{\sqrt{2}i}\sigma_3} \lambda^{\frac{\gamma_1}{\sqrt{2}i}\sigma_3} (1-\lambda)^{\frac{\gamma_2}{\sqrt{2}i}\sigma_3}, & \Im \lambda<0, \\
|s|^{\frac{\gamma_1+\gamma_2}{\sqrt{2}i}\sigma_3} \lambda^{\frac{\gamma_1}{\sqrt{2}i}\sigma_3} (1-\lambda)^{\frac{\gamma_2}{\sqrt{2}i}\sigma_3}\sigma_3\sigma_1, & \Im \lambda>0.
\end{array}
\right.
\end{equation}
In the above formula,  the principal branches are chosen such that $\arg \lambda, \arg(1-\lambda) \in (-\pi,\pi)$.
\end{lemma}
\begin{proof}
First, we prove that $E_{N, y}$ is analytic in $B(x_1, \delta)$. From its explicit expression, we only need to establish the analyticity of $E_{N, y}(z)$ on the interval $[-1,1] \cap B(x_1, \delta)$. When $ x \in (x_1-\delta, x_1+\delta)$, it is readily seen that $E_{N, y}(x)_{-}^{-1} E_{N, y}(x)_+ = I$. Given that $z \to x_1$ as $\lambda_y(z) \to 0$, and that $\gamma_1$ is real, the factor $\lambda^{\frac{\gamma_1}{\sqrt{2}i}\sigma_3} $ in \eqref{hatpsiinfty} stays bounded as $z \to x_1$. This yields $E_{N, y}(z) = O(1)$ as $z\to x_1$, revealing that $z = x_1$ is a removable singularity. An analogous argument works for $z=x_2$. Consequently, $E_{N, y}(z)$ is analytic in $B(x_1, \delta)$.

Next, using \eqref{jump hatPsi} and the analyticity of $E_{N, y}$, one can easily verify that $P_{x_1}(z)$ defined in \eqref{Pmerge} satisfies the same jump conditions as $S(z)$ in $B(x_1, \delta)$. 

It remains to check the matching condition \eqref{eq: p-p-infty-math}. From \eqref{Pmerge} and \eqref{Emerge}, it follows that
\begin{equation*}
P_{x_1}(z) P^{\infty}(z)^{-1}  =  E_{N, y} (z) \widehat{\Phi}_{PV}(\lambda_y(z); s_{N, y})e^{-N\xi(z)\sigma_3} \widehat{\Phi}^\infty(\lambda_y(z))^{-1} e^{-N\xi_+(x_1)\sigma_3} E_{N, y}(z)^{-1}.
\end{equation*}
For $\Im z < 0$, the diagonality of  $\widehat{\Phi}^{\infty}(\lambda(z))$ (see \eqref{hatpsiinfty}) ensures that $e^{-N\xi(z)\sigma_3} \widehat{\Phi}^\infty(\lambda_y(z))^{-1} e^{-N\xi_+(x_1)\sigma_3}$ mutually commute. Moreover, it follows from \eqref{xi(z)} and \eqref{lambdayz} that
\begin{eqnarray*}
e^{-N\xi(z)\sigma_3} e^{-N\xi_+(x_1)\sigma_3} = e^{-N\xi(z)\sigma_3} e^{N\xi_-(x_1)\sigma_3} = e^{\frac{s_{N, y}}{2} \lambda_y(z) \sigma_3}.
\end{eqnarray*}
Using the above two formulas, we get, for $\Im z < 0$,
\begin{eqnarray}
P_{x_1}(z) P^{\infty}(z)^{-1} 
= E_{N, y} (z) \widehat{\Phi}_{PV}(\lambda_y(z); s_{N, y}) e^{\frac{s_{N, y}}{2} \lambda_y(z) \sigma_3} \widehat{\Phi}^\infty(\lambda_y(z))^{-1} E_{N, y}(z)^{-1}.
\end{eqnarray}
For $\Im z > 0$, using the fact $\sigma_1 a^{\sigma_3} = a^{-\sigma_3} \sigma_1$, a similar computation gives us
\begin{eqnarray}
P_{x_1}(z) P^{\infty}(z)^{-1} = E_{N, y} (z) \widehat{\Phi}_{PV}(\lambda_y(z); s_{N, y}) e^{-\frac{s_{N, y}}{2} \lambda_y(z) \sigma_3} \widehat{\Phi}^\infty(\lambda_y(z))^{-1} E_{N, y}(z)^{-1}.
\end{eqnarray}
To approximate $\widehat{\Phi}_{PV}(\lambda_y(z); s_{N, y}) e^{\pm \frac{s_{N, y}}{2} \lambda_y(z) \sigma_3} \widehat{\Phi}^\infty(\lambda_y(z))^{-1}$ in the formulas above, we adopt an argument similar to that in \cite[Lemma 7.1]{CFL2021} to obtain
\begin{equation}
\widehat{\Phi}_{PV}(\lambda_y(z); s_{N, y}) e^{\mp \frac{s_{N, y}}{2} \lambda_y(z) \sigma_3} \widehat{\Phi}^\infty(\lambda_y(z))^{-1} = I+O((N\delta)^{-1}), \qquad N \to \infty,
\end{equation}
uniformly for $z\in \partial B(x_1, \delta)$, with $0<y< \delta/2$ and $x_1$ in a fixed compact subset of $(-1, 1)$. Regarding the behavior of $E_{N, y}(z)$ for $z\in \partial B(x_1, \delta)$, both $N$-dependent terms, $\widehat{\Phi}^\infty(\lambda_y(z))$ and $e^{-N \xi_+(x_1) \sigma_3}$, are bounded as $N \to \infty$. From \eqref{hatpsiinfty}, we know that $\widehat{\Phi}^\infty(\lambda_y(z))$ is bounded in $\partial B(x_1, \delta)$. In addition, $P^\infty (z)$ and $w_L(z)$ have no singularities on $\partial B(x_1, \delta)$.By the definition in \eqref{Emerge}, we have $E_{N, y}(z)^{\pm 1} = O(1)$ as $N \to \infty$, uniformly for $z \in \partial B(x_1, \delta)$. Combining this with the formulas above yields the required matching condition \eqref{eq: p-p-infty-math}.

We have thus completed the proof.
\end{proof}

Having constructed all the global and local parametrices, we now give the final transformation as follows:
\begin{eqnarray} \label{merging-R-def}
R(z) = \left\{\begin{array}{ll}
                S(z)P_{-1}(z)^{-1} & z\in B(-1, \delta), \\
                S(z)P_{1}(z)^{-1} & z\in B(1, \delta), \\
                S(z)P_{x_1}(z)^{-1} & z\in B(x_1, \delta), \\
                S(z)P^\infty(z)^{-1} & z\in\mathbb{C}\setminus (B(x_1, \delta)\cup B(-1, \delta)\cup B(1, \delta)).
              \end{array}
\right.
\end{eqnarray}
One can readily verify that $R(z)$ satisfies the following RH problem.

\begin{rhp}
\hfill
\begin{itemize}
\item[(a)] $R: \mathbb{C}\setminus \Sigma_R \to \mathbb{C}^{2\times 2}$ is analytic, where 
\begin{equation}
\Sigma_R = \Sigma_S \cup \partial B(\pm 1, \delta) \cup \partial B(x_1, \delta) \setminus \{  [-1, 1] \cup B(\pm 1, \delta) \cup B(x_1, \delta)\}.
\end{equation}
The orientations on $\partial B(\pm 1, \delta)$ and $\partial B(x_1, \delta)$ are taken to be clockwise.

\item[(b)] $R_+(z) = R_-(z)J_R(z)$, where
\begin{eqnarray}
J_R(z) = \left\{\begin{array}{ll}
                P_{-1}(z)P^\infty(z)^{-1} & z\in  \partial B(- 1, \delta), \\
                P_{1}(z)P^\infty(z)^{-1} & z\in  \partial B(1, \delta), \\
                P_{x_1}(z)P^\infty(z)^{-1} & z\in \partial B(x_1, \delta), \\
                P^\infty(z)J_S(z)P^\infty(z)^{-1} & z\in\Sigma_R\setminus (B(x_1, \delta)\cup B(-1, \delta)\cup B(1, \delta)).
              \end{array}
\right.
\end{eqnarray}
\item[(c)] As $z\to\infty$, we have
\begin{equation}
R(z) = I+O(z^{-1}).
\end{equation}
\end{itemize}
\end{rhp}

From \eqref{jumps for S}, one can see that, there exists a positive constant $c>0$ such that
\begin{equation*}
J_{R}(z) =  I+O\left(e^{-2c N}\right), \qquad z\in\Sigma_R\setminus (B(x_1, \delta)\cup B(-1, \delta)\cup B(1, \delta)).
\end{equation*}
For $z \in \partial B(x_1, \delta)\cup \partial B(\pm 1, \delta)$, it follows from \eqref{sec3.3-matching1} and \eqref{eq: p-p-infty-math} that
\begin{eqnarray}
J_{R}(z) =  I+O\left(\frac{1}{N \delta}\right), \quad z \in \partial B(x_1, \delta),  \qquad J_{R}(z) = I+O\left(\frac{1}{N}\right), \quad z \in \partial B(\pm 1, \delta).
\end{eqnarray}
Then, by the standard result for small-norm RH problems (see \cite[Section 7]{DeiftZhou}), we have
\begin{eqnarray} \label{Asym for R in merging case}
R(z) =  I + O\left(\frac{1}{N(|z|+1)}\right) \quad \textrm{and} \quad  R'(z) = O\left(\frac{1}{N(|z|+1)}\right), \quad \text{ as } N \to \infty,
\end{eqnarray}
uniformly for $z \in \mathbb{C} \setminus \Sigma_R$.

\subsection{Local Parametrix near $-1$}

This section is devoted to the construction of the local parametrix for Case (III) depicted in Figure \ref{fig:contours 3}. To begin, we seek a function $P_{1}(z)$ satisfying the following RH problem in $B(1, \delta)$.

\begin{rhp} \label{rhp for local parametrix near 1 in the edge case}
\hfill
\begin{itemize}
\item[(a)] $P_{1}: B(1, \delta) \setminus ([-1, \infty) \cup \Gamma_2 \cup \overline{\Gamma}_2)\to \mathbb{C}^{2\times 2}$ is analytic.
\item[(b)] $P_{1}$ satisfies the same jump relations as $S$ on $B(1, \delta)\cap ([-1, \infty) \cup \Gamma_1 \cup \overline{\Gamma}_1)$.
\item[(c)] For $z\in \partial B(1, \delta)$, we have the matching condition
\begin{equation}\label{matching condition: P-1Pinfty}
P_{1}(z)P^\infty(z)^{-1} = I+O(1/N), \qquad \textrm{as } N\to\infty.
\end{equation}
\item[(d)] As $z\to 1$, $P_{1}(z)$ satisfies the same asymptotic behaviours as $S$.
\end{itemize}
\end{rhp}

The local parametrix near $z=1$  is provided by the Airy parametrix. We now turn to the local parametrix near $z=-1$, which contains the jump point $x$. A function $P_{-1}(z)$ is sought such that it satisfies the following RH problem in $B(-1, \delta)$.

\begin{rhp} \label{rhp of the local parametrix for the edge regime near -1}
\hfill
\begin{itemize}
\item[(a)] $P_{-1}: B(-1, \delta) \setminus ([-1, \infty) \cup \Gamma_1 \cup \overline{\Gamma}_1 \cup  \Gamma_2 \cup \overline{\Gamma}_2)\to \mathbb{C}^{2\times 2}$ is analytic.
\item[(b)] $P_{-1}$ satisfies the same jump relations as $S$ on $B(-1, \delta)\cap ([-1, \infty) \cup \Gamma_1 \cup \overline{\Gamma}_1  \cup \Gamma_2 \cup \overline{\Gamma}_2)$.
\item[(c)] For $z\in \partial B(-1, \delta)$, we have the matching condition
\begin{equation} \label{mathcing condition near -1}
P_{-1}(z)P^\infty(z)^{-1} = I+O(1/N), \qquad \text{ as } N\to\infty.
\end{equation}
\item[(d)] As $z\to -1$, $P_{-1}(z)$ remains bounded.
\end{itemize}
\end{rhp}

The above RH problem can be solved by first constructing an auxiliary model RH problem.

\begin{figure}[htbp]
    
\centering

\tikzset{every picture/.style={line width=0.75pt}} 

\begin{tikzpicture}[x=0.75pt,y=0.75pt,yscale=-1,xscale=1]

\draw    (100,105) -- (250.33,104.8) -- (378.33,105.8) ;
\draw [shift={(180.17,104.89)}, rotate = 179.92] [fill={rgb, 255:red, 0; green, 0; blue, 0 }  ][line width=0.08]  [draw opacity=0] (8.93,-4.29) -- (0,0) -- (8.93,4.29) -- cycle    ;
\draw [shift={(319.33,105.34)}, rotate = 180.45] [fill={rgb, 255:red, 0; green, 0; blue, 0 }  ][line width=0.08]  [draw opacity=0] (8.93,-4.29) -- (0,0) -- (8.93,4.29) -- cycle    ;
\draw [shift={(100,105)}, rotate = 359.92] [color={rgb, 255:red, 0; green, 0; blue, 0 }  ][fill={rgb, 255:red, 0; green, 0; blue, 0 }  ][line width=0.75]      (0, 0) circle [x radius= 3.35, y radius= 3.35]   ;
\draw    (100,105) .. controls (135.33,160.8) and (217.33,161.8) .. (250.33,104.8) ;
\draw [shift={(181.26,147.07)}, rotate = 178.81] [fill={rgb, 255:red, 0; green, 0; blue, 0 }  ][line width=0.08]  [draw opacity=0] (8.93,-4.29) -- (0,0) -- (8.93,4.29) -- cycle    ;
\draw    (100,105) .. controls (136.33,47.8) and (216.33,46.8) .. (250.33,104.8) ;
\draw [shift={(181.22,61.83)}, rotate = 181.14] [fill={rgb, 255:red, 0; green, 0; blue, 0 }  ][line width=0.08]  [draw opacity=0] (8.93,-4.29) -- (0,0) -- (8.93,4.29) -- cycle    ;
\draw    (250.33,104.8) -- (344.33,210.8) ;
\draw [shift={(300.65,161.54)}, rotate = 228.43] [fill={rgb, 255:red, 0; green, 0; blue, 0 }  ][line width=0.08]  [draw opacity=0] (8.93,-4.29) -- (0,0) -- (8.93,4.29) -- cycle    ;
\draw    (250.33,104.8) -- (347.33,19.8) ;
\draw [shift={(302.59,59.01)}, rotate = 138.77] [fill={rgb, 255:red, 0; green, 0; blue, 0 }  ][line width=0.08]  [draw opacity=0] (8.93,-4.29) -- (0,0) -- (8.93,4.29) -- cycle    ;

\draw (88,110) node [anchor=north west][inner sep=0.75pt]   [align=left] {0};
\draw (247,116) node [anchor=north west][inner sep=0.75pt]   [align=left] {1};
\draw (156,167) node [anchor=north west][inner sep=0.75pt]   [align=left] {$\displaystyle \overline{\Sigma _{\phi ,\ 1}}$};
\draw (151,23) node [anchor=north west][inner sep=0.75pt]   [align=left] {$\displaystyle \Sigma _{\phi ,\ 1}$};
\draw (334,37) node [anchor=north west][inner sep=0.75pt]   [align=left] {$\displaystyle \Sigma _{\phi ,\ 2}$};
\draw (326,163) node [anchor=north west][inner sep=0.75pt]   [align=left] {$\displaystyle \overline{\Sigma _{\phi ,\ 2}}$};

\end{tikzpicture}

\caption{The jump contours for the model RH problem for $\Phi$}
\end{figure}
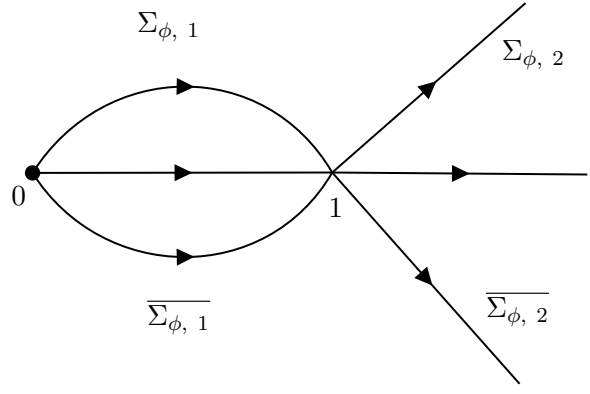

\begin{rhp} \label{model-rhp-hard-edge}
\hfill
\begin{itemize}
\item[(a)] $\Phi = \Phi(\cdot; u)$ is analytic on $\mathbb{C}\setminus ([0, \infty)\cup \Sigma_{\phi, 1}\cup \Sigma_{\phi, 2}\cup \overline{\Sigma_{\phi, 1}}\cup \overline{\Sigma_{\phi, 2}})$. The contours are depicted in the figure.

\item[(b)] On $[0, \infty)\cup \Sigma_{\phi, 1}\cup \Sigma_{\phi, 2}\cup \overline{\Sigma_{\phi, 1}}\cup \overline{\Sigma_{\phi, 2}}\setminus \{0, 1\}$,  $\Phi$ satisfies the following jump conditions:
    \begin{eqnarray}
    \begin{array}{ll}
    \Phi_+(\lambda) = \Phi_-(\lambda)\left(\begin{array}{ll} 1 & 0\\ e^{-\sqrt{2}\pi\gamma}e^{-4(-\lambda u)^{1/2}}e^{-\alpha\pi i} & 1
     \end{array}\right) & \text{for } \lambda\in\Sigma_{\Phi, 1},\\
     \Phi_+(\lambda) = \Phi_-(\lambda)\left(\begin{array}{ll} 1 & 0\\ e^{-\sqrt{2}\pi\gamma} e^{-4(-\lambda u)^{1/2}}e^{\alpha\pi i} & 1
     \end{array}\right) & \text{for } \lambda\in\overline{\Sigma_{\Phi, 1}},\\
     \Phi_+(\lambda) = \Phi_-(\lambda)\left(\begin{array}{ll} 1 & 0\\ e^{-4(-\lambda u)^{1/2} }e^{-\alpha\pi i} & 1
     \end{array}\right) & \text{for } \lambda\in\Sigma_{\Phi, 2},\\
     \Phi_+(\lambda) = \Phi_-(\lambda)\left(\begin{array}{ll} 1 & 0\\ e^{-4(-\lambda u)^{1/2} }e^{\alpha\pi i} & 1
     \end{array}\right) & \text{for } \lambda\in\overline{\Sigma_{\Phi, 2}},\\
     \Phi_+(\lambda) = \Phi_-(\lambda)\left(\begin{array}{ll} 0 & e^{\sqrt{2}\pi \gamma}\\ -e^{-\sqrt{2}\pi \gamma} & 0
     \end{array}\right) & \text{for } \lambda\in (0, 1),\\
     \Phi_+(\lambda) = \Phi_-(\lambda)\left(\begin{array}{ll} 0 & 1\\ -1 & 0
     \end{array}\right) & \text{for } \lambda\in(1, \infty),
    \end{array}
    \end{eqnarray}
    The principal branches are chosen in the square roots, and $u>0$ and $\gamma\in\mathbb{R}$ are parameters.

\item[(c)] As $\lambda\to\infty$,
\begin{equation}\label{Phiasym}
\Phi(\lambda) = (I+O(\lambda^{-1}))\left(\begin{array}{cc} 1 & -\sqrt{2}\gamma\\ 0 & 1\end{array}\right)(-\lambda)^{\frac{1}{4}\sigma_3} B 
\end{equation}
where the principal branches are chosen, and
\begin{equation}\label{A}
B = \frac{1}{\sqrt 2}\left(\begin{array}{cc} 1 & i\\ i & 1\end{array}\right).
\end{equation}

\item[(d)] As $\lambda\to 1$, $\Phi(\lambda) = O(\log |\lambda-1|)$.

\item[(e)] For $\alpha<0$, $\Phi(\lambda)$ has the following behaviour as $\lambda\to 0$:
\begin{eqnarray}
\Phi(\lambda) = O\left(
                   \begin{array}{cc}
                     |\lambda|^{\alpha / 2} & |\lambda|^{\alpha / 2} \\
                     |\lambda|^{\alpha / 2} & |\lambda|^{\alpha / 2} \\
                   \end{array}
                 \right), \text{ as } \lambda\to 0.
\end{eqnarray}
For $\alpha=0$, $\Phi(\lambda)$ has the following behaviour as $\lambda\to 0$:
\begin{eqnarray}
\Phi(\lambda) = O\left(
                   \begin{array}{cc}
                     \log |\lambda| & \log |\lambda| \\
                     \log |\lambda| & \log |\lambda| \\
                   \end{array}
                 \right), \text{ as } \lambda\to 0.
\end{eqnarray}
For $\alpha>0$, $\Phi(\lambda)$ has the following behaviour as $\lambda\to 0$:
\begin{eqnarray}
\Phi(\lambda) = \left\{\begin{array}{ll}
O\left(
   \begin{array}{cc}
     |\lambda|^{\alpha/2} & |\lambda|^{-\alpha/2} \\
     |\lambda|^{\alpha/2} & |\lambda|^{-\alpha/2} \\
   \end{array}
 \right), & \text{as } \lambda\to 0 \text{ outside the lens,}\\
O\left(
   \begin{array}{cc}
     |\lambda|^{-\alpha/2} & |\lambda|^{-\alpha/2} \\
     |\lambda|^{-\alpha/2} & |\lambda|^{-\alpha/2} \\
   \end{array}
 \right), & \text{as } \lambda\to 0 \text{ inside the lens,}
\end{array}
\right.
\end{eqnarray}
\end{itemize}
\end{rhp}

As stated in Section \ref{Appendix:subsec:modelrhp-hardedge}, or more specifically, in \cite[3.5.1]{Dai-Lu-JUE}, the solution to this model RH problem exists as $u\to\infty$, although the directions of the contours are reversed.

Now we are at the stage of solving the RH problem \ref{rhp of the local parametrix for the edge regime near -1}. We first define
\begin{eqnarray}\label{f}
f(z) = -\frac{1}{4}(\pi i - \xi(z))^2,
\end{eqnarray}
where $\xi(z)$ is given in \eqref{xi(z)}. A direct computation (see also \cite[(5.59)]{CG2021} yields
\begin{eqnarray} \label{eq:f-expan}
f(z) = 2(z+1) - \frac{1}{3}(z+1)^2 + O((z+1)^3), \quad \text{ as } z\to -1, 
\end{eqnarray}
which shows that $f(z)$ is a conformal map in a neighborhood of $z=-1$. Denote
\begin{equation} \label{def:un-x}
u_{N, x} = N^2 f(x)_+ = -\frac{N^2}{4} (\pi i - \xi_+(x))^2,
\end{equation}
where $x$ is the jump point in Case (III); see Figure \ref{fig:contours 3} for an illustration. Obviously, we have 
\begin{equation}\label{u1}
u_{N, x} = O(N^2(1+x)), \quad \text{ as } N\to\infty,
\end{equation}
uniformly for $x\in (-1+N^{-2}\log\log N, -1+ \varepsilon)$, so that $u_{N, x} \to \infty$ as $N\to \infty$ for $x$ in this regime.

We further define
\begin{equation}\label{lambdadef}
\lambda_{x}(z) = \frac{f(z)}{f(x)},
\end{equation}
which is also a conformal map in a neighborhood of $z = -1$, mapping $z=-1$ to $0$ and $z = x$ to $1$, respectively. In addition, the following approximation holds:
\begin{equation}\label{lambda}
\lambda_{x}(z) = O((1+x)^{-1}), \quad \text{ as } x\to -1,
\end{equation}
uniformly for $z\in\partial B(-1, \delta)$.  We also define 
\begin{equation} \label{W(z)}
W(z) = (-z - 1)^{\alpha / 2}, \quad \text{ for } z\in \mathbb{C} \setminus [-1, \infty),
\end{equation}
 with the branch of the square root taken so that $W(z)>0$ for $z < -1$. This definition then yields
\begin{eqnarray} \label{eq:W-wj-relation}
W^2(z) = \left\{\begin{array}{ll}
                  e^{-\alpha\pi i} (z+1)^\alpha, & \Im z>0; \\
                  e^{\alpha\pi i} (z+1)^\alpha, & \Im z<0.
                \end{array}
\right.
\end{eqnarray}


The jump contour $\Sigma_S$ of $S$ near $-1$ is now chosen in such a way that $\lambda_x$ maps $\Sigma_S\cap B(-1, \delta)$ to the jump contour $\Sigma_{\Phi,1}$ and $\Sigma_{\Phi, 2}$ appearing in RH Problem \ref{model-rhp-hard-edge}. The following lemma provides the solution to RH Problem \ref{rhp of the local parametrix for the edge regime near -1}.

\begin{lemma} \label{lemma: local parametrix near -1}
Let $\lambda_x(z)$ and $W(z)$ be defined in \eqref{lambdadef} and \eqref{W(z)},  and $\Phi$ be the solution to the model RH problem \ref{model-rhp-hard-edge}. Then, the solution to RH problem \ref{rhp of the local parametrix for the edge regime near -1} is given by
\begin{equation}\label{P}
P_{-1}(z) = E_{-1}(z)\Phi(\lambda_x(z), u_{N, x}) W(z)^{-\sigma_3},
\end{equation}
where  $E_{-1}(z)$ is an analytic function in $B(-1, \delta)$ defined as
\begin{equation}\label{E}
E_{-1}(z) = P^{\infty}(z) W(z)^{\sigma_3} M(\lambda_x(z))^{-1},
\end{equation}
with $P^{\infty}(z)$ given in \eqref{Pinfty} and $M(\lambda)$ given by
\begin{equation}\label{def of tilde M}
 M(\lambda) = (-\lambda)^{\sigma_3/4} B \left(\frac{1+e^{-\pi i / 2}\sqrt{-1/\lambda}}{\sqrt{-\frac{1}{\lambda}+1}}\right)^{-i\sqrt{2}\gamma \sigma_3},\quad \lambda \in \mathbb{C} \setminus (0, \infty),
\end{equation}
where 
\begin{equation}\label{A}
B = \frac{1}{\sqrt 2}\left(\begin{array}{cc} 1 & i\\ i & 1\end{array}\right).
\end{equation}
\end{lemma}

\begin{proof}
The proof is similar to \cite[Sec 9.1]{CFL2021} and \cite[Lemma 3.17]{Dai-Lu-JUE}. We omit the details here.
\end{proof}

With all the global and local parametrices constructed, we define the final transformation as follows:
\begin{eqnarray}\label{transform map: edge S to R}
R(z) = \left\{\begin{array}{ll} S(z)P^\infty(z)^{-1} & \text{ for } z\in \mathbb{C}\setminus (\Gamma_S\cup B(1, \delta)\cup B(-1, \delta)),\\
S(z)P_{1}(z)^{-1} & \text{ for } z\in B(1, \delta), \\
S(z)P_{-1}(z)^{-1} & \text{ for } z\in B(-1, \delta).
\end{array}\right.
\end{eqnarray}
Then, $R(z)$ satisfies the following RH problem.

\begin{rhp}
\hfill
\begin{itemize}
\item[(a)] $R: \mathbb{C}\setminus \Sigma_R \to \mathbb{C}$ is analytic, where
\begin{equation}
 \Sigma_R = \Sigma_S \cup \partial B(\pm 1, \delta) \setminus \{[-1, \infty) \cup B(\pm 1, \delta)\} 
\end{equation}
and the orientations on $\partial B(\pm 1, \delta)$ are taken to be clockwise.
\item[(b)] On $\Sigma_R$, $R$ satisfies the following jump conditions
\begin{eqnarray}
R_+(z) = 
\left\{
\begin{array}{ll}
R_{-}(z) P^\infty(z)
    \left(
      \begin{array}{cc}
        1 & 0 \\
        e^{-\sqrt{2}\pi\gamma} (z+1)^{-\alpha} e^{-2N\xi(z)} & 1 \\
      \end{array}
    \right) (P^\infty(z))^{-1}, & z\in \Gamma_2 \cup \overline{\Gamma}_{2}, \\
R_{-}(z) P_{-1}(z)P^{\infty}(z)^{-1}, & z\in B(-1, \delta),\\
R_{-}(z) P_1(z)P^{\infty}(z)^{-1}, & z\in B(1, \delta),\\
R_{-}(z) P^\infty(z) \left(\begin{array}{ll}
1 & (x+1)^\alpha e^{2N\xi(x)}\\
0 & 1
\end{array}
\right) (P^\infty(z))^{-1} & x>1.
\end{array}
\right.
\end{eqnarray}
\item[(c)] As $z\to\infty$,
\begin{equation}
R(z) = I+O(z^{-1}).
\end{equation}
\end{itemize}
\end{rhp}

With the definition of $\xi(z)$ given in \eqref{xi(z)} and the matching relations in  \eqref{matching condition: P-1Pinfty} and \eqref{mathcing condition near -1} taken into account, one can readily verify that the jump matrices on $\Gamma_2$, $\overline{\Gamma}_2$ and $(1, \infty)$ are $I + O(e^{-N})$ as $N\to\infty$, while those on $\partial B(-1, \delta) \cup \partial B(1, \delta)$ are $I + O(1/N)$ uniformly for $-1 + N^{-2} \log\log N < x < -1 + \varepsilon$. Consequently, the standard small-norm RH problem directly implies 
\begin{equation}\label{Rasymp}
R(z) = I+O(N^{-1}) \quad \textrm{and} \quad R'(z) = O(N^{-1}), \qquad \textrm{as } N\to\infty, 
\end{equation}
uniformly for $z\in \mathbb{C}\setminus \Gamma_R$ and $-1 + N^{-2}\log\log N < x < -1 + \varepsilon$. Here $\varepsilon > 0$ is a constant that can be sufficiently small.

\subsection{Local Parametrix near $1$}

In this subsection, we construct the local parametrix near $1$, i.e., the Case (IV) illustrated in Figure \ref{fig:contours 4}, where the jump point $x$ is close to $1$. First, we consider the following RH problem in $B(-1, \delta)$.

\begin{rhp} \label{rhp for local parametrix near -1 in the edge case}
\hfill
\begin{itemize}
\item[(a)] $P_{-1}: B(-1, \delta) \setminus ([-1, 1] \cup \Gamma_1 \cup \overline{\Gamma}_1)\to \mathbb{C}^{2\times 2}$ is analytic.
\item[(b)] $P_{-1}$ satisfies the same jump relations as $S$ on $B(-1, \delta)\cap ([-1, 1] \cup \Gamma_1 \cup \overline{\Gamma}_1)$.
\item[(c)] For $z\in \partial B(-1, \delta)$, we have the matching condition
\begin{equation}\label{matching condition: P-1Pinfty-soft}
P_{-1}(z)P^\infty(z)^{-1} = I+O(1/N), \qquad \textrm{as } N\to\infty.
\end{equation}

\item[(d)] As $z\to - 1$, $P_{-1}(z)$ has the same asymptotic behaviours as $S(z)$.
\end{itemize}
\end{rhp}

As in RH problem \ref{sec3.4-RHP1}, the local parametrix near $-1$  is given by the Bessel parametrix. 
Next, we discuss the local parametrix near  $z=1$, which encloses the jump point $x$. The corresponding RH problem is as follows.

\begin{rhp} \label{rhp of the local parametrix for the edge regime near 1}
\hfill
\begin{itemize}
\item[(a)] $P_{1}: B(1, \delta) \setminus ([-1, \infty) \cup \Gamma_1 \cup \overline{\Gamma}_1 \cup \Gamma_2 \cup \overline{\Gamma}_2)\to \mathbb{C}^{2\times 2}$ is analytic.
\item[(b)] $P_{1}$ satisfies the same jump relations as $S$ on $B(1, \delta)\cap ([-1, 1] \cup \Gamma_1 \cup \overline{\Gamma}_1 \cup \Gamma_2 \cup \overline{\Gamma}_2)$.
\item[(c)] For $z\in \partial B(1, \delta)$, we have the matching condition
\begin{equation} \label{mathcing condition near 1}
P_{1}(z)P^\infty(z)^{-1} = I+O(N^{-1/3}), \qquad \text{ as } N\to\infty.
\end{equation}

\item[(d)] As $z\to 1$, $P_{1}(z)$ has the same asymptotic behaviours as $S(z)$.

\end{itemize}
\end{rhp}

Next, we solve this RH problem by means of the model RH problem in Appendix \ref{Appendix:subsec:modelrhp-softedge}. 

We first define
\begin{eqnarray}\label{f-soft}
f(z) = \big(-\frac{3}{2}\xi(z)\big)^{2/3},
\end{eqnarray}
where $\xi(z)$ is given in \eqref{xi(z)}. A direct computation (see also \cite[(5.51)]{CG2021} yields
\begin{eqnarray} \label{eq:f-expan-soft}
f(z) = 2^{-1/3} (z-1) - \frac{10}{2^{1/3}}(z-1)^2 + O((z-1)^3), \quad \text{ as } z\to 1,
\end{eqnarray}
which shows that $f(z)$ is a conformal map in a neighborhood of $z=1$. Denote
\begin{equation} \label{def:un-x-soft}
u_{N, x} = -N^{2/3} f(x) = \big(\frac{3N}{2}|\xi_+(x)|\big)^{2/3},
\end{equation}
where $x$ is the jump point in Case (IV); see Figure \ref{fig:contours 4} for an illustration. Obviously, we have 
\begin{equation}\label{u}
u_{N, x} = O(N^{2/3}(1-x)), \quad \text{ as } N\to\infty,
\end{equation}
uniformly for $x\in (1-\varepsilon, 1- N^{-2/3}\log \log N)$, so that $u_{N, x} \to \infty$ as $N\to \infty$ for $x$ in this regime.

Define
\begin{equation}\label{lambdadef-soft}
\lambda_{x}(z) = -\frac{f(z)}{f(x)},
\end{equation}
which is also a conformal map near $z = 1$, mapping $z=1$ and $z = x$ to $0$ and $-1$, respectively. Furthermore, the following approximation holds.
\begin{equation}\label{lambda-soft}
\lambda_{x}(z) = O((1-x)^{-1}), \quad \text{ as } x\to 1,
\end{equation}
uniformly for $z\in\partial B(1, \delta)$. 

The jump contour $\Sigma_S$ of $S$ near $1$  is now chosen so that its intersection with $ B(1, \delta)$ is mapped by $\lambda_x$ to the jump contours $\Sigma_{\Phi,1}$ and $\Sigma_{\Phi, 2}$ appearing in RH problem \ref{model-rhp-soft-edge}. Then the following lemma provides the solution to RH problem \ref{rhp of the local parametrix for the edge regime near 1}


\begin{lemma}
Let $\lambda_x(z)$ be defined in \eqref{lambdadef-soft},  and $\Phi = \Phi_{\text{soft}}$ be the solution to the model RH problem \ref{model-rhp-soft-edge}. Then, the solution to RH problem \ref{rhp of the local parametrix for the edge regime near 1} is given by
\begin{equation}\label{P-soft}
P_{1}(z) = E_{1}(z)\Phi(\lambda_x(z), u_{N, x}) (z-1)^{-\frac{\alpha}{2}\sigma_3},
\end{equation}
where  $E_{1}(z)$ is an analytic function in $B(1, \delta)$ defined as
\begin{equation}\label{E-soft}
E_{1}(z) = P^{\infty}(z) (z+1)^{\frac{\alpha}{2}\sigma_3} M_1(\lambda_x(z))^{-1},
\end{equation}
with 
\begin{equation}\label{def of M}
M_1(\lambda) = \lambda^{-\frac{1}{4}\sigma_3} B\left(\frac{1+e^{-\pi i/2}\sqrt{\lambda}}{\sqrt{\lambda+1}}\right)^{-i\sqrt{2}\gamma \sigma_3}, \quad \lambda \in \mathbb{C}\setminus (-\infty, 0],
\end{equation}
and $P^\infty$ defined in \eqref{Pinfty}.
\end{lemma}

\begin{proof}
We first need to establish the analyticity of $E_{1}(z)$ in $B(1, \delta)$.  In view of its expression, it suffices to verify the analyticity of $E_{1}(z)$ across $[-1, \infty)\cap B(1, \delta)$. Since  $P^\infty(z)$ and $M_1(\lambda_x(z))$ obey the same jump conditions in $B(1, \delta)$, it follows that $E_{1}(z)_{-}^{-1}E_{1}(z)_{+} = I$ for $z\in (1-\delta, 1+\delta)$. We also note that $\lambda_x(z) \to 0$ as $z\to 1$, it follows that $E_{1}(z)$ is of order $O((z-1)^{-1/4})$; as $z\to x$, $E_{1}(z)$ is of order $O(1)$. Thus $x$ and $1$ are removable singularities, so $E_{1}(z)$ is analytic in $B(1, \delta)$. 

Then, using RH problem \ref{model-rhp-soft-edge}, it is not difficult to see that $P_{1}(z)$ satisfies the same jump conditions as $S(z)$ in $B(-1, \delta)$.

It therefore remains to establish that $P_{1}(z)$ fulfills the matching condition  \eqref{mathcing condition near 1}. It is almost the same as that of Lemma \ref{lemma: local parametrix near -1}, so we omit the details. This concludes the proof of the lemma.

\end{proof}

Now that all the global and local parametrices have been constructed, we proceed to define the final transformation by
\begin{eqnarray}\label{transform map: edge S to R}
R(z) = \left\{\begin{array}{ll} S(z)P^\infty(z)^{-1} & \text{ for } z\in \mathbb{C}\setminus (\Gamma_S\cup B(1, \delta)\cup B(-1, \delta)),\\
S(z)P_{1}(z)^{-1} & \text{ for } z\in B(1, \delta), \\
S(z)P_{-1}(z)^{-1} & \text{ for } z\in B(-1, \delta),
\end{array}\right.
\end{eqnarray}
which satisfies the following RH problem.

\begin{rhp}
\hfill
\begin{itemize}
\item[(a)] $R: \mathbb{C}\setminus \Sigma_R \to \mathbb{C}$ is analytic, where
\begin{equation}
\Sigma_R = \Sigma_S \cup \partial B(\pm 1, \delta) \setminus \{[-1, \infty) \cup B(\pm 1, \delta)\}
\end{equation}
and the orientations on $\partial B(\pm 1, \delta)$ are taken to be clockwise.
\item[(b)] On $\Sigma_R$, $R$ satisfies the following jump conditions
\begin{eqnarray}
R_+(z) = 
\left\{
\begin{array}{ll}
R_{-}(z) P^{\infty}(z)
    \left(
      \begin{array}{cc}
        1 & 0 \\
        e^{-\sqrt{2}\pi\gamma} (z+1)^{-\alpha} e^{-2N\xi(z)} & 1 \\
      \end{array}
    \right) (P^\infty(z))^{-1}, & z\in \Gamma_1 \cup \overline{\Gamma}_1, \\
R_{-}(z) P_{-1}(z)P^{\infty}(z)^{-1}, & z\in B(-1, \delta),\\
R_{-}(z) P_1(z)P^{\infty}(z)^{-1}, & z\in B(1, \delta),\\
R_{-}(z) P^{\infty}(z) \left(\begin{array}{ll}
1 & (x+1)^\alpha e^{2N\xi(x)}\\
0 & 1
\end{array}
\right) (P^\infty(z))^{-1}  & x>1.
\end{array}
\right.
\end{eqnarray}
\item[(c)] As $z\to\infty$,
\begin{equation}
R(z) = I+O(z^{-1}).
\end{equation}
\end{itemize}
\end{rhp}

Using the definition of $\xi(z)$ in \eqref{xi(z)} together with \eqref{matching condition: P-1Pinfty-soft} and \eqref{mathcing condition near 1}, we know that as $N\to\infty$, the jumps on $\Gamma_1$, $\overline{\Gamma}_1$ and $(1, \infty)$ are $I + O(e^{-N})$ as $N\to\infty$; while the jumps on $\partial B(-1, \delta) \cup \partial B(1, \delta)$ are uniformly $I + O(N^{-1/3})$ in the range $1 - \varepsilon < x < 1- N^{-2/3}\log\log N$. By the standard theory of small-norm RH problems, it follows that
\begin{equation}\label{Rasymp}
R(z) = I+O(N^{-1/3}) \quad \textrm{and} \quad R'(z) = O(N^{-1/3}), \qquad \textrm{as } N\to\infty, 
\end{equation}
uniformly for $z\in \mathbb{C}\setminus \Gamma_R$ and $1 - \varepsilon < x < 1 - N^{-2/3}\log\log N$, where $\varepsilon > 0$ is a sufficiently small constant.

\section{Asymptotics of Hankel Determinants} \label{Sec:Asy-hankle}


Below, we will use the steepest descent analysis and the differential identities provided in Proposition \ref{differential identities of Hankel determinants} to obtain the asymptotic formulae for the corresponding Hankel determinants. These results in  \ref{Asymptotics in the Separated Regime}, \ref{Asymptotics in the merging regime}, and \ref{Asymptotics in the edge} may be compared with Theorems 1.8, 1.9, and 1.10 in \cite{CFL2021}, respectively.

\subsection{Asymptotics of Hankel Determinants in the Separated Regime}

We begin with Case (I), in which the weight function has jump discontinuities at $x_1$  and $x_2$, neither close to each other nor to the endpoints $\pm 1$. The asymptotics of the corresponding Hankel determinants follow from \cite[Theorem 1.2]{CG2021}, which applies to a family of Laguerre-type weight functions with finitely many Fisher-Hartwig singularities in $(-1,1)$. The weight function in \cite[Theorem 1.2]{CG2021} is required to be analytic in a fixed neighborhood of $[-1, \infty)$ except for these Fisher-Hartwig singularities. We need a slight generalization of their result to allow the weight to possess finitely many additional singularities in the complex plane that may be close to $[-1,1]$. More precisely, we consider weight functions that are analytic and uniformly bounded on the following domain:
\begin{eqnarray}
  \mathcal{S}_N &=& \left\{ z\in \mathbb{C}: |\Re z| \leq 1-3\delta_N, |\Im z| < \varepsilon_N / 2\right\} \nonumber\\
& & \cup \left\{ z\in\mathbb{C}: \Re z \geq 1-3\delta_N, |\Im z| < 3\delta_N \right\}\nonumber\\
& & \cup \left\{ z\in\mathbb{C}: -1-3\delta_N \leq \Re z \leq -1+3\delta_N, |\Im z| < 3\delta_N \right\}. 
\end{eqnarray}
where $\varepsilon_N = N^{-1+\alpha}$ and $\delta_N = N^{-\alpha / 2}$ with $0<\alpha<2/3$.

\begin{proposition}\label{Asymptotics in the Separated Regime}
Let $\gamma_1, \gamma_2 \in \mathbb{R}$ and  $-1<x_1<x_2<1$. Assume that $w=w_N$ is a sequence of functions that are real analytic on $[-1, \infty)$ and uniformly bounded on $\mathcal{S}_N$. Then, we have
\begin{eqnarray} \label{equation: Asymptotics in seperated regime}
& & \log\frac{D_N(x_1, x_2; \gamma_1, \gamma_2; w_N)}{D_N(x_1, x_2; \gamma_1, \gamma_2; 0)} \nonumber\\
&=& N\int \frac{w_N(x) \sqrt{1-x}}{\pi \sqrt{1+x}}dx + \frac{\alpha}{2\pi} \int_{-1}^1 \frac{w_N(x)}{\sqrt{1-x^2}}dx - \frac{\alpha}{2}w_N(-1) \nonumber\\
& & + \frac{1}{2}\sigma^2(w_N) + \sum_{j=1}^2 \frac{\gamma_j}{\sqrt{2}} \sqrt{1-x_j^2} \mathcal{U} w_N(x_j) + o(1)
\end{eqnarray}
as $N\to\infty$. Here $\mathcal{U}w_N$ is the finite Hilbert transform defined by
\begin{equation} \label{def-Hilbert-transform}
    \mathcal{U}w_N(x) = \frac{1}{\pi} P. V. \int_{-1}^1 \frac{w_N(t)}{x-t}\frac{dt}{\sqrt{1-t^2}},
\end{equation}
and
\begin{equation}
\sigma^2(w; t) = \frac{1}{2\pi^2} \int\int_{[-1,1]^2} w'(x)t'(y) \Sigma(x, y)dxdy =-\frac{1}{2\pi} \int_{-1}^1 w'(t) \mathcal{U} t(s) \sqrt{1-s^2}ds, \label{eq:sigma-def}
\end{equation}
with $\sigma^2(w) := \sigma^2(w; w)$ and $\Sigma(x, y)$ given in \eqref{correlation kernel of X}.
\end{proposition}

\begin{proof}

Taking the parameters in \cite[Theorem 1.2]{CG2021} to be  $m=2, \alpha_0 = \alpha, \alpha_1 = \alpha_2 = 0; V(x)=2(x+1); \psi_V(x) = 1/ \pi$ and $t_k = x_k, \beta_k = -\frac{i\sqrt{2}\gamma_k}{2}$ for $k=1, 2$, one immediately gets the asymptotics. The formula in \eqref{equation: Asymptotics in seperated regime} is subsequently deduced by contrasting the asymptotics for the two weight functions $W(x)=w_N(x)$ and $W(x) = 0$, with $W(x)$ being the weight function in \cite[Theorem 1.2]{CG2021}.

According to \cite[Theorem 1.2]{CG2021}, when $w_N$ is analytic in a fixed neighborhood of $[-1, \infty)$, the error term in \eqref{equation: Asymptotics in seperated regime} is $O(\frac{\log N}{N})$. In our situation, however, we merely assume that $w_N$ is analytic and uniformly bounded in a shrinking neighborhood $\mathcal{S}_N$ of $[-1, \infty)$. Therefore, by employing a proof analogous to that of \cite[Prop. 7.5]{CFL2021}, we can prove that the error term deteriorates slightly and decays to $o(1)$ as $N\to\infty$.
\end{proof}

\subsection{Asymptotics of the Hankel Determinants in the Merging Regime}


The asymptotic behaviour of the Hankel determinants in the merging regime is described below for Case (II), where the two jumps in the weight function are close together but not near $\pm 1$. These results are obtained by following the methods of \cite[Sec. 7.4-7.5]{CFL2021} and \cite[Sec. 4.2]{Dai-Lu-JUE}, and are stated in the proposition that follows; the proofs are not included here.
\begin{proposition}\label{Asymptotics in the merging regime}
Let $\gamma_1, \gamma_2\in\mathbb{R}$, $-1<x_1<x_2<1$, $t(x)$ be real analytic on $[-1,1]$, and $D_N(x_1, x_2; \gamma_1, \gamma_2; t)$ be defined in \eqref{eq1}. Then, as $N\to\infty$, we have
\begin{eqnarray}
 \log \frac{D_N(x_1, x_2; \gamma_1, \gamma_2; t)}{D_N(x_1; \gamma_1+\gamma_2; t)} = 
\sqrt{2} \pi N\gamma_2 \int_{x_1}^{x_2}\psi_V(u) \sqrt{\frac{1-u}{1+u}}du - \gamma_1\gamma_2\max\{0, \log |x_1-x_2|N\} + O(1),
\end{eqnarray}
where the error term is uniform for $-1+\delta<x_1<x_2<1-\delta, 0<x_2-x_1<\delta$ for $\delta$ sufficiently small.
\end{proposition}

\subsection{Asymptotics of Hankel Determinants near $\pm 1$}


Finally, we consider Case (III) and (IV), in which the weight function has only one jump point located near one of the endpoints $1$ or $-1$.  By employing the methods from \cite[Sec 4.3]{Dai-Lu-JUE} for Case (III) and \cite[Sec. 9.2]{CFL2021} for Case (IV), we obtain the following proposition.

\begin{proposition}\label{Asymptotics in the edge}
Let $\gamma\in [-M, M]$ with $M > 0$ be a positive constant, $w(x)$ be real analytic on $[-1,1]$, and $D_N(x; \gamma; w)$ be defined in \eqref{eq2}. Then, as $N\to\infty$, we have
\begin{equation}
\log \frac{D_N(x; \gamma; 0)}{D_N(x; 0; 0)} = \sqrt{2}\gamma N\int_{-1}^x \sqrt{\frac{1-s}{1+s}}ds+ \frac{\gamma^2}{2}\log N+\frac{\gamma^2}{4}\log (1+x) + O(1), \label{eq: aysmptotics near -1}
\end{equation}
where the error term is uniform for $ -1 + N^{-2}\log\log N \leq x \leq 1 - \varepsilon$. Similarly, as $N\to\infty$, we have
\begin{equation}
\log \frac{D_N(x; \gamma; 0)}{D_N(x; 0; 0)} = \sqrt{2}\gamma N\int_{-1}^x \sqrt{\frac{1-s}{1+s}}ds+ \frac{\gamma^2}{2}\log N+\frac{3\gamma^2}{4}\log (1-x) + O(1), \label{eq: aysmptotics near 1}
\end{equation}
where the error term is uniform for $-1+\varepsilon \leq x \leq  1 - N^{-2/3}\log\log N$.
\end{proposition}

\begin{remark}
The difference in the coefficients of the logarithmic terms between \eqref{eq: aysmptotics near -1} and \eqref{eq: aysmptotics near 1} arises from the fact that the equilibrium measure behaves as $O((1+x)^{-1/2})$ near $-1$,  whereas its integral is of order $O((1+x)^{1/2})$.
On the other hand, near $1$ the equilibrium measure behaves as $O((1-x)^{1/2})$, hence the integral is of order $O((1-x)^{3/2})$. This fact accounts for the three-time difference of the logarithmic term in the proposition.
\end{remark}

\section{Eigenvalue Rigidity} \label{Eigenvalue Rigidity}


It follows from \eqref{Heine identity} that the exponential moments are related to the Hankel determinants. In the previous section, we presented the asymptotic expressions for the Hankel determinants, which can be used to prove that the random measure $d\mu_N^\gamma(x) $ in \eqref{eq:dmu-N} converges to a GMC measure as $N \to \infty$.

\subsection{Estimation on the Exponential Moments of $h_N(x)$}


To begin with, we verify that $h_N$ in \eqref{definition of hN} satisfies \cite[Assumption 3.1]{CFL2021}. More precisely, we prove the proposition below.

\begin{proposition}\label{Assumption 3.1}
Let $A$ be any compact set in $(-1, 1)$ with positive Lebesgue measure. For any $\gamma>0$ and sufficiently large $N$, there exists a constant $C_{\gamma, A} > 0$ such that
\begin{equation} \label{prop5-1-eq1}
\mathbb{E} [ e^{\gamma h_N(x)}] \geq C_{\gamma, A} N^{\gamma^2 / 2} \qquad \textrm{ for all } x\in A.
\end{equation}
Moreover, there exists a constant $R_{\gamma} > 0$ such that 
\begin{equation} \label{prop5-1-eq2}
\mathbb{E} [ e^{\gamma h_N(x)}] \leq R_{\gamma} N^{\gamma^2 / 2} \qquad \textrm{for } x \in (-1, 1).
\end{equation}
The same bounds hold for $-h_N(x)$.
\end{proposition}

\begin{proof}
Since the case of $-h_N(x)$ is similar to that of $h_N(x)$, it suffices to prove the estimates for $h_N(x)$.

First, we compute the exponential moment of $h_N(x)$. By definition \eqref{definition of hN},
\begin{eqnarray}
\mathbb{E} [e^{\gamma h_N(x)}] = \mathbb{E} \big[e^{\sqrt{2}\pi\gamma \sum_{j=1}^N 1_{\lambda_j\leq x} - \sqrt{2}\pi\gamma N F(x)}\big] 
= e^{-\sqrt{2}\pi \gamma N F(x)} \mathbb{E} \big[e^{\sqrt{2}\pi\gamma \sum_{j=1}^N 1_{\lambda_j\leq x}}\big],
\end{eqnarray}
where $F(x)$ is the distribution function defined in \eqref{distribution of eq measure}. By \eqref{Heine identity}, we get
\begin{eqnarray}
\mathbb{E} [e^{\gamma h_N(x)}] = e^{-\sqrt{2}\pi \gamma N F(x)} \frac{N!}{Z_N} D_N (x; \gamma; 0)  = e^{-\sqrt{2}\pi \gamma N F(x)} \frac{D_N (x; \gamma; 0)}{D_N (0; 0; 0)}.  \label{EV-Heine Identity1}
\end{eqnarray}
Then, from Proposition \ref{Asymptotics in the edge}, we can see that as $N\to\infty$,
\begin{equation}\label{EehN}
\log \mathbb{E} [e^{\gamma h_N(x)}] = \frac{\gamma^2}{2}\log N + O(1),
\end{equation}
 uniformly for $|x|\leq 1- \varepsilon$, where $\varepsilon > 0$ is a constant independent of $N$. This implies that for any compact set $A\subset (-1, 1)$, the lower bound \eqref{prop5-1-eq1} holds.

We now turn to the upper bound \eqref{prop5-1-eq2}. For $x \in [-1+\varepsilon, 1-\varepsilon]$ with $\varepsilon > 0$, the result is immediate from Proposition \ref{Asymptotics in the edge} and \eqref{EV-Heine Identity1}. For $x$ close to $-1$, from the definition of $F(x)$ in \eqref{distribution of eq measure}, there exists $C > 0$ such that for $x\leq -1 + \varepsilon$, $(1+x)^{1/2} \leq C F(x)$. Then, by \eqref{eq: aysmptotics near -1} and \eqref{EV-Heine Identity1}, there exists $R_\gamma > 0$ such that for $x \geq -1 + N^{-2}\log\log N$,
\begin{equation} \label{hN moment bulk}
\mathbb{E}[e^{\gamma h_N(x)}] \leq R_\gamma (NF(x))^{\gamma^2/2}. 
\end{equation}
Noting that $F(x) \in [0, 1]$, we have the right-hand side bounded by $R_\gamma N^{\gamma^2/2}$. Next, this estimate is carried over to $-1  < x\leq -1 + N^{-2}\log\log N$. Making use of $h_N(x)$ in \eqref{definition of hN}, we arrive at
\begin{multline}
 h_N(-1+N^{-2}\log\log N) + \sqrt{2}\pi N F(-1+N^{-2}\log \log N) \nonumber\\
= \sqrt{2}\pi \sum_{1\leq j \leq N} 1_{\lambda_j \leq -1 + N^{-2}\log \log N}  \geq \sqrt{2}\pi \sum_{1\leq j \leq N} 1_{\lambda_j \leq x} 
\geq  h_N(x)
\end{multline}
for all $ x\in  (-1 , -1 + N^{-2}\log\log N]$. This gives us 
\begin{equation}
\mathbb{E} [ e^{\gamma h_N(x)}] \leq e^{\sqrt{2}\pi \gamma N  F(-1+N^{-2}\log \log N) } \, \mathbb{E} [e^{\gamma h_N(-1+N^{-2}\log \log N)}].
\end{equation}
Note that $F(-1+N^{-2} \log \log N) \leq C' N^{-1}(\log\log N)^{1/2}$ for some $C'>0$. Applying \eqref{hN moment bulk} at $x = -1+N^{-2} \log \log N$, we obtain
\begin{eqnarray} \label{hN moment edge -1}
\mathbb{E} [e^{\gamma h_N(x)}] &\leq& e^{\sqrt{2}\pi\gamma C' (\log\log N)^{1/2}} R_{\gamma} \Big(N F(-1+N^{-2} \log \log N) \Big)^{\gamma^2 / 2} \nonumber\\
&\leq& e^{\sqrt{2}\pi\gamma C' (\log\log N)^{1/2}} R_{\gamma} \Big( C'(\log\log N)^{1/2} \Big)^{\gamma^2 / 2},
\end{eqnarray}
which is of order $o(N^{\gamma^2/2})$. 

Next, we turn to the case where $x$ is close to $1$. Similarly, there exists $C > 0$ such that for $x \geq 1 - \varepsilon$,  $(1-x)^{3/2} \leq C F(x)$. Then, by \eqref{eq: aysmptotics near 1} and \eqref{EV-Heine Identity1}, there exists $R’_\gamma > 0$ such that for $x \leq 1 - N^{-2/3}\log\log N$,
\begin{equation} \label{hN moment bulk right}
\mathbb{E}[e^{\gamma h_N(x)}] \leq R'_\gamma (NF(x))^{\gamma^2/2},
\end{equation}
which is exactly \eqref{hN moment bulk}. To extend the result to $1 - N^{-2/3}\log\log N \leq x  < 1$, first noting that from the definition of $h_N(x)$ in \eqref{definition of hN}, one has
\begin{equation}
h_N(x) \leq \sqrt{2}\pi N \left(1 - F(1-N^{-2/3}\log\log N) \right).
\end{equation}
Hence for $x\in [1-N^{-2/3}\log\log N, 1)$, one has
\begin{equation} \label{hN moment edge 1}
\mathbb{E}[e^{\gamma h_N(x)}] \leq \mathbb{E} [e^{\sqrt{2}\pi \gamma N \big(1 - F(1-N^{-2/3}\log\log N)\big)}]  \leq e^{C (\log\log N)^{3/2}}
\end{equation}
for some $C>0$, which is also of order $o(N^{\gamma^2/2})$. 

Combining \eqref{hN moment bulk}, \eqref{hN moment edge -1}, \eqref{hN moment bulk right} and \eqref{hN moment edge 1}, we conclude the proof of this proposition.
\end{proof}

\subsection{Proof of Theorem \ref{maximum of hN}: The Maximum of $h_N(x)$}

In this section, we are at the stage of verifying that $h_N(x)$ satisfies all the requirements in \cite[Assumption 2.5]{CFL2021}, hence proving the convergence of $d\mu_N$ constructed in \eqref{eq:dmu-N} to the GMC measure $d\mu^\gamma$. With Proposition \ref{Asymptotics in the Separated Regime} and Proposition \ref{Asymptotics in the merging regime} in hand, it is almost the same as \cite[Section 2.7.2]{CFL2021} and \cite[Section 5.1]{Dai-Lu-JUE}. So we omit the details and give the following proposition.

In this section, we establish that 
$h_N(x)$ fulfills all the requirements of \cite[Assumption 2.5]{CFL2021} by invoking Proposition \ref{Asymptotics in the Separated Regime} and Proposition \ref{Asymptotics in the merging regime}. This naturally leads to the proof of the convergence of 
$d\mu^\gamma$, as defined in \eqref{eq:dmu-N}, to the GMC measure 
$d\mu^\gamma$. Further details of the proof can be found in \cite[Section 2.7.2]{CFL2021} and \cite[Section 5.1]{Dai-Lu-JUE}.

\begin{proposition} \label{prop-GMC-convergence}
Let $0<\gamma<\sqrt{2}$, the eigenvalue counting function $h_N(x)$ be defined in \eqref{definition of hN}, and $X(x)$ be the log-correlated Gaussian field with correlation kernel \eqref{correlation kernel of X}. Then, the measure $d\mu_N^\gamma(x)$ in \eqref{eq:dmu-N} converges to the GMC measure $d\mu^\gamma(x) = \frac{e^{\gamma X(x)}}{\mathbb{E} [e^{\gamma X(x)} ] }$ as $N \to \infty$ in law, with respect to the weak topology. 
\end{proposition}

\begin{proof}
Since $h_N(x)$ defined in \eqref{definition of hN} is not centered (e.g. see \cite[Corollary 2.2]{CG2021}, we introduce the following modified eigenvalue counting function
\begin{equation} \label{tilde hN}
\tilde{h}_N(x) = h_N(x) - \sqrt{2}\pi \left(\frac{\alpha}{2\pi} \int_{-1}^x \frac{ds}{\sqrt{1-s^2}} - \frac{\alpha}{2}1_{[1, \infty)}(x)\right),
\end{equation}
where we keep the indicator function to make the symbols conform with those in \cite{Dai-Lu-JUE}. Note that $\tilde{h}_N(x)$ differs from $h_N(x)$ by a deterministic quantity.

Adopting the standard arguments of the estimations for the exponential moments of $\tilde{h}_N(x)$, for instance, see \cite[Sec. 2.7.2]{CFL2021} or \cite[Prop. 5.2]{Dai-Lu-JUE}, one can see that $\tilde{h}_N(x)$ satisfies all the conditions in \cite[Assumption 2.5]{CFL2021}, hence 
\begin{equation}
d\mu_N^\gamma = \frac{e^{\gamma h_N(x)}}{\mathbb{E}e^{\gamma h_N(x)}} = \frac{e^{\gamma \tilde{h}_N(x)}}{\mathbb{E}e^{\gamma \tilde{h}_N(x)}}
\end{equation}
converges in law to the GMC measure $d\mu^\gamma$ associated with the log-correlated field $X(x)$ with correlation kernel \eqref{correlation kernel of X}.
\end{proof}

The lower bound for the maximum of $h_N(x)$ is obtained via the convergence of $d\mu_N^\gamma(x)$ to the GMC measure.

\begin{proposition} \label{lower bound of hN}
Let $h_N(x)$ be the eigenvalue counting function defined in \eqref{definition of hN}. Then, for any compact set $A\subset (-1, 1)$ with a positive Lebesgue measure, we have, for any $\delta>0$,
\begin{equation} \label{eq:hn-lowerbound}
\lim_{N\to\infty} \mathbb{P} \left[ \max_{ x \in A} h_N(x) \geq (\sqrt{2}-\delta) \log N\right] = 1.
\end{equation}
\end{proposition}

\begin{proof}
Applying Propositions \ref{Assumption 3.1} and \ref{prop-GMC-convergence}, we confirm that $h_N(x)$ fulfills the hypotheses of \cite[Theorem 3.4]{CFL2021}; this directly yields the lower bound \eqref{eq:hn-lowerbound}.
\end{proof}

\begin{remark}
The lower bound for the maximum of $h_N(x)$ given by the above Proposition is much stronger than Theorem \ref{maximum of hN}, as it implies that the lower bound of $\max |h_N(x)|$ holds in probability on any compact subset of $[-1, 1]$ with positive Lebesgue measure, rather than merely on the entire interval.
\end{remark}

The upper bound of $h_N(x)$ follows from a same argument by adopting the Markov inequality as \cite[Prop. 5.6]{Dai-Lu-JUE}, hence we only state the result and omit the proof.
\begin{proposition} \label{prop: upper bound of hN}
Let $h_N(x)$ be the eigenvalue counting function defined in \eqref{definition of hN}. Then, we have, for any $\delta>0$, 
\begin{equation} \label{eq:hn-upperbound}
\lim_{N\to\infty} \mathbb{P} \left[ \max_{ x \in [-1, 1]} h_N(x) \leq (\sqrt{2}+\delta) \log N\right] = 1.
\end{equation}
\end{proposition}

Combining the above two propositions, we conclude the proof of Theorem \ref{maximum of hN}.

\subsection{Proof of Theorem \ref{main theorem}: The Eigenvalue Rigidity}

We now proceed to the proof of our main result, Theorem \ref{main theorem}. In this subsection, we will assume $\psi_V(x) = \frac{1}{\pi}$ for simplicity, which is also the coefficient for a standard LUE. We emphasize that no generality is lost since, by our assumption of $\psi_V(x)$, there exist $c$ and $C$ such that $c\leq \psi_V(x) \leq C$ in $[-1, \infty)$. Recalling the definition of $h_N(x)$ from \eqref{definition of hN}, we have
\begin{equation} \label{hN(lambda)}
h_N(\lambda_j)  = \sqrt{2} \pi N (F(\kappa_j) - F(\lambda_j)).
\end{equation}
When $ j \asymp N$ and $ N-j \asymp N$, the above formula yields
\begin{equation} \label{eq: lower-bound-approx}
h_N(\lambda_j)  = \sqrt{2}\pi N F'(\kappa_j)(\kappa_j - \lambda_j) +  O(1).
\end{equation}
Then, the lower bound in \eqref{eq:main-theorem} follows directly from Proposition \ref{lower bound of hN}.  The rest of this section is therefore dedicated to proving the upper bound in \eqref{eq:main-theorem}. 

We claim that it suffices to consider the case where $\lambda_j > \kappa_j$. Otherwise, if $\lambda_j \leq \kappa_j$, we have
\begin{equation} \label{eq 1: only need one side}
|F(\kappa_j) - F(\lambda_j)| = |F'(\xi_j)|(\kappa_j - \lambda_j), \qquad \xi_j\in (\lambda_j, \kappa_j).
\end{equation}
Since $F'(x)>0$ is  monotonically decreasing  in $[-1, 1]$ (cf. \eqref{distribution of eq measure} and \eqref{LUE eq measure}), it follows that
\begin{equation}\label{eq 2: only need one side}
F'(\kappa_j) (\kappa_j - \lambda_j) \leq F'(\xi_j)(\kappa_j - \lambda_j). 
\end{equation}
Hence, the result follows directly from Prop \ref{prop: upper bound of hN}. Thus, we only need to consider the case $\lambda_j>\kappa_j$ for the remainder of the proof.

We will divide the proof into two parts: the regime where $\kappa_j < 0$ and the regime where $\kappa_j > 0$. Noting that for the $j$-s such that $\kappa_j$ is near $0$, we must have $j \asymp N $ and $ N-j \asymp N$, hence \eqref{eq: lower-bound-approx} must hold. As a direct corollary of Proposition \ref{prop: upper bound of hN}, the upper bound in Theorem \ref{main theorem} must hold. Hence we do not need to consider the subtle case where $\kappa_j$ and $\lambda_j$ are of opposite signs and only need to focus on the case that either $j$ or $N-j$ is of order $o(N)$. In what follows, we denote by $C, C'$ positive constants independent of $N$, not necessarily equal in value.

\subsubsection{Eigenvalues near the Hard Edge: $\kappa_j \leq 0$}
Define $N_1$ as the largest integer $j$ such that $\kappa_j \leq 0$. The proof is divided into two parts: we first use an iteration method to push the regime where the upper bound of Theorem \ref{main theorem} holds as close to $-1$ as possible. Then, we refine the estimation of $h_N(x)$ near $-1$ in order to prove the hard edge rigidity.

\smallskip
\noindent\textbf{Bulk Rigidity: The Method of Iteration} 
\smallskip

As mentioned in the above discussion, we first obtain the upper bound in the bulk, and then extend it to the vicinity of the endpoint $-1$.

\begin{proposition} \label{Main theorem: Eigenvalue Rigidity on the Bulk: left}
For any $\delta > 0$ and $F(x)$ defined in \eqref{distribution of eq measure}, we have
\begin{equation} \label{eq: prop5.7-formula}
\lim_{N\to\infty}\mathbb{P}\left(F'(\kappa_j)|\lambda_j - \kappa_j| \leq \frac{(1+\varepsilon)\log N}{\pi N} \text{ for } j = N^\delta, \cdots, N_1 \right) = 1.
\end{equation}
\end{proposition}

\begin{proof}

The proof is similar to that of \cite[Prop 5.7]{Dai-Lu-JUE} since in both cases $F'(x) \asymp (1+x)^{-1/2}$. Hence, we omit the details.

\end{proof}

\smallskip
\noindent\textbf{Hard Edge Rigidity: Refinement of the Bound of $h_N(x)$} 
\smallskip


In this section, we give the proof of the upper bound estimate \eqref{eq:main-theorem} near the edge. As preparation, we need the following two lemmas concerning the bounds of the eigenvalue counting function $h_N(x)$. Since $F'(x) \asymp (1+x)^{-1/2}$ as $x\to -1$ the proofs are similar to those of Lemmas 5.8 and 5.9 in \cite{Dai-Lu-JUE}; hence, we omit the proofs and refer to \cite{Dai-Lu-JUE}.

\begin{lemma} \label{hN near edge first lemma}
Let $h_N(x)$ be the eigenvalue counting function defined in \eqref{definition of hN}. Then, for any constant $C>0$, there exists a sufficiently small $\delta>0$ such that
\begin{equation}
\lim_{N\to\infty} \mathbb{P} [\max_{x\leq \kappa_{2N^{\delta}}} h_N(x) \geq C\log N] = 0.
\end{equation}
By symmetry, similar estimates hold for $-h_N(x)$.
\end{lemma}

As $x$ approaches the edge $-1$, the approximation in the above lemma admits further improved.

\begin{lemma} \label{hN near edge second lemma}
Let $h_N(x)$ be the eigenvalue counting function defined in \eqref{definition of hN}. Then, for any constant $C>0$, we have
\begin{equation}
\lim_{N\to\infty} \mathbb{P} \left[\max_{x\leq \kappa_{2\log N}} h_N(x) \geq C (\log N)^{1/2}\right] = 0.
\end{equation}
By symmetry, similar estimates hold for $-h_N(x)$.
\end{lemma}

Using the two lemmas above, we now prove the following proposition concerning eigenvalue rigidity near the edge.

\begin{proposition} \label{main proposition for edge rigidity}
For $F(x)$ defined in \eqref{distribution of eq measure}, there exists $\delta > 0$ such that
\begin{equation} \label{eq: main proposition for edge rigidity}
\lim_{N\to\infty} \mathbb{P} \left(F'(\kappa_j)|\lambda_j-\kappa_j| \leq \frac{(1+\varepsilon)\log N}{\pi N} \text{ for } 1\leq j \leq N^\delta \right) = 1.
\end{equation}
\end{proposition}

\begin{proof}

From the discussion in \eqref{eq 1: only need one side} and \eqref{eq 2: only need one side}, we only need to consider the case $\lambda_j>\kappa_j$. By defition \eqref{distribution of eq measure}, we have  
\begin{equation}
F'(\kappa_j) = \frac{1}{\pi} \sqrt{\frac{1-\kappa_j}{1+\kappa_j}}  \sim \frac{4N}{j \pi^2}
\end{equation}
as $N\to\infty$. Therefore, it suffices to show that there exists some $\delta>0$ such that
\begin{equation}\label{edge: Prob: lambda-kappa}
 \lim_{N\to\infty} \mathbb{P} \left(\lambda_j-\kappa_j\leq \frac{j\log N}{4N^2} \quad\textrm{ for } j\leq N^\delta\right) = 1.
\end{equation}

Next, we proceed with the proof using the above two lemmas.

\textbf{Step 1: Consider the regime where $\log N\leq j\leq N^\delta$.} Assume that there exists $j$ such that $\lambda_j - \kappa_j > \frac{j\log N}{4N^2}$. From the definition of $h_N(x)$ in \eqref{definition of hN}, we have
\begin{equation}\label{same}
\frac{1}{\sqrt{2}\pi N} h_N\left(\kappa_j + \frac{j\log N}{4N^2}\right) \leq  F(\kappa_j) - F\left(\kappa_j+\frac{j\log N}{4N^2}\right)
\end{equation}
for sufficiently large $N$. Recall that
\begin{equation}\label{same2}
\kappa_j + \frac{j\log N}{4N^2} = -1+C\frac{j^2}{N^2}+\frac{j\log N}{4N^2} +  O(j^4/N^4)
\end{equation}
as $N\to\infty$. Observe that for $\log N \leq j \leq N^\delta$, the expansion is dominated by the $j^2/N^2$ term. Applying the mean value theorem together with the behavior of $F(x)$ near $-1$, we can find a constant $C > 0$ such that
\begin{equation} \label{same3}
F\left(\kappa_j+\frac{j\log N}{4N^2}\right) - F(\kappa_j) \geq  C\cdot \frac{N}{j}\cdot \frac{j\log N}{4N^2}  C\frac{\log N}{4N}.
\end{equation}
Substituting this into \eqref{same} yields
\begin{equation}\label{star}
\frac{1}{\sqrt{2}\pi N} h_N\left(\kappa_j + \frac{j\log N}{4N^2}\right) \leq -C\frac{\log N}{N}.
\end{equation}
As $\kappa_j +1 \asymp \frac{j^2}{N^2}$, we have $
\kappa_j +\frac{j\log N}{4N^2} \leq \kappa_{2j}\leq \kappa_{2N^\delta}$, which implies that $\min_{x\leq \kappa_{2N^\delta}} h_N(x) \leq -C\log N$ for some $C>0$.  Using Lemma \ref{hN near edge first lemma}, one can readily see that the probability of this event tends to $0$. Consequently, for sufficiently small $\delta>0$, inequality \eqref{edge: Prob: lambda-kappa} is valid for  $\log N\leq j\leq N^\delta$.

\textbf{Step 2: Consider the regime where $1\leq j\leq \log N$.} Again, assume that there exists $j$ such that $\lambda_j - \kappa_j > \frac{j\log N}{4N^2}$. The formulas \eqref{same} and \eqref{same2} still hold. However, since $1\leq j\leq \log N$, the term $\frac{j\log N}{4N^2}$ in expansion \eqref{same2} dominates over $\frac{j^2}{N^2}$. Consequently, the approximation in \eqref{same3} is modified to
\begin{equation}
F\left(\kappa_j+\frac{j\log N}{4N^2}\right) - F(\kappa_j) \geq C\cdot  \frac{2N}{\sqrt{j\log N}}\cdot \frac{j\log N}{4N^2} = C\frac{\sqrt{j\log N}}{2N}.
\end{equation}
Substituting this result into \eqref{same} yields
\begin{equation}\label{star2}
\frac{1}{\sqrt{2}\pi N} h_N\left(\kappa_j + \frac{j\log N}{4N^2}\right) \leq -C\frac{\sqrt{j\log N}}{4N} \leq -C\frac{ \sqrt{\log N}}{4N}.
\end{equation}
From \eqref{same2}, we have $
\kappa_j + \frac{j\log N}{4N^2} \leq \kappa_{2\log N}
$
for sufficiently large $N$. Therefore, inequality \eqref{star2} implies that $\min_{x\leq \kappa_{2\log N}} h_N(x) \leq -C \sqrt{\log N}$.  Lemma \ref{hN near edge second lemma} implies that the probability of this event tends to $0$ as $N \to \infty$, which in turn yields that  \eqref{edge: Prob: lambda-kappa} holds for $1\leq j\leq \log N$.

\smallskip

We have now established \eqref{edge: Prob: lambda-kappa} which completes the proof of Proposition \ref{main proposition for edge rigidity}.
\end{proof}

This completes the proof for the regime where $\kappa_j \leq 0$. Now let us turn to the right regime where $\kappa_j \geq 0$, and in this region the soft edge displays a very different pattern from the hard edge. 

\subsubsection{Eigenvalues near the Soft Edge: $\kappa_j > 0$}

We will now prove Theorem \ref{main theorem} in the regime $\kappa_j > 0$, or equivalently, $j > N_1$. Similarly, we first prove that the upper bound of \eqref{eq:main-theorem} holds in the bulk region which can be extended to very close to the soft edge $1$. Then, we use the fact that $F'(1) = 0$ and the monotonicity of $F'(x)$ together to prove the edge rigidity.

\smallskip
\noindent\textbf{Bulk Rigidity} 
\smallskip

Different from the regime where $\kappa_j \leq 0$, we do not have an apriori estimate of $|\lambda_j - \kappa_j|$;  since from \eqref{LUE eq measure}, it can be expected that $|\lambda_j - \kappa_j|$ will be larger and larger as we approach $1$, but this will be ``balanced'' by the decrease of $F'(x)$. Hence, the method of iteration does not work. Instead, we first prove a weaker estimation of $|\lambda_j - \kappa_j|$ in the bulk and close to the edge, i.e.,  push the result to the regime where $N-j \sim o(N)$. Then we argue that it suffices for our proof of bulk rigidity.

\begin{lemma} \label{lemma: bulk estimation of lambda-kappa near 1}
For $N_1< j \leq N- [\log N]^4$, one has
\begin{equation}\label{eq: bulk estimation of lambda-kappa near 1}
\lim_{N\to\infty} \mathbb{P} \left( \lambda_j - \kappa_j \leq N^{-2/3} \right) = 1.
\end{equation}
\end{lemma}

\begin{proof}
The idea of the proof is also similar to Proposition \ref{main proposition for edge rigidity}. Assume that there exists $N_1<j \leq N - [\log N]^4$ such that $\lambda_j - \kappa_j > N^{-2/3}$. From the definition of $h_N(x)$ in \eqref{definition of hN}, one has
\begin{equation}
\frac{1}{\sqrt{2}\pi N} h_N\left(\kappa_j + N^{-2/3}\right) \leq F(\kappa_j) - F\left(\kappa_j + N^{-2/3}\right)
\end{equation}
for sufficiently large $N$. Since $F'(x)$ decreases in $(-1, 1)$, we have
\begin{eqnarray}\label{estimation of hN near 1: bulk}
& & \frac{1}{\sqrt{2}\pi N} h_N\left(\kappa_j + N^{-2/3}\right)  \leq  F(\kappa_{N-[\log N]^4}) - F\left( \kappa_{N-[\log N]^4} + N^{-2/3} \right) 
\end{eqnarray}
But from the definition of $\kappa_j$, we have
\begin{equation}
1 - \kappa_{N-[\log N]^4} \asymp \left(\frac{(\log N)^4}{N} \right)^{2/3}
\end{equation}
Using \eqref{estimation of hN near 1: bulk}, one has
\begin{eqnarray}
& & \frac{1}{\sqrt{2}\pi N} h_N\left(\kappa_j + N^{-2/3}\right)  \leq -C N^{-2/3} \cdot \left(\frac{(\log N)^4}{N} \right)^{1/3} = -C (\log N)^{4/3} / N,
\end{eqnarray}
which further implies that
\begin{equation}
h_N(\kappa_j + N^{-2/3}) \leq -C (\log N)^{4/3} 
\end{equation}
for some $N_1 < j \leq N - [\log N]^4$. By Theorem \ref{maximum of hN}, the limit probability of this event tends to $0$ as $N\to\infty$.
\end{proof}

Now we can present our result for bulk rigidity in the right-half regime.

\begin{proposition} \label{Main theorem: Eigenvalue Rigidity on the Bulk: right}
For $F(x)$ defined in \eqref{distribution of eq measure}, we have
\begin{equation} \label{eq: Eigenvalue Rigidity on the Bulk: right}
\lim_{N\to\infty}\mathbb{P}\left(F'(\kappa_j)|\lambda_j - \kappa_j| \leq \frac{(1+\varepsilon)\log N}{\pi N} \text{ for } j = N_1, \cdots, N - [\log N]^4 \right) = 1.
\end{equation}
\end{proposition}

\begin{proof}
Recalling the Taylor expansion of $h_N(x)$ in \eqref{Taylor expansion of hN} and the maximum of $h_N(x)$ in Theorem \ref{maximum of hN}, it suffices to show the second-order term in \eqref{Taylor expansion of hN} is at most of order $O(1)$. In fact, one can see that
\begin{eqnarray}
\zeta_j \leq \lambda_j \leq \lambda_{N - [\log N]^4} \leq \kappa_{N - [\log N]^4} + N^{-2/3} \leq 1- C (\log N)^{8/3} N^{-2/3} + N^{-2/3}
\end{eqnarray}
holds in probability as $N\to\infty$. It follows that
\begin{eqnarray}
\sqrt{2}\pi N \frac{|F''(\zeta_j)|}{2}(\lambda_j - \kappa_j)^2 \leq C N (\log N)^{-4/3} N^{1/3} N^{-4/3} \leq C'(\log N)^{-4/3},
\end{eqnarray}
which is of order $o(1)$. In the first inequality, we have used Lemma \ref{lemma: bulk estimation of lambda-kappa near 1}. Therefore, we conclude the proof of the proposition.

\end{proof}

\smallskip
\noindent\textbf{Soft Edge Rigidity: Refinement of the Bound of $h_N(x)$} 
\smallskip

Now we will turn to proving the soft edge rigidity, i.e., the regime where $j\geq N - [\log N]^4$. Similarly, we will show that the upper bound of $h_N(x)$ can also be refined near the soft edge $1$. To be specific, we have the following lemma.
\begin{lemma} \label{hN near soft edge first lemma}
Let $h_N(x)$ be the eigenvalue counting function defined in \eqref{definition of hN}. Then, for any constant $C>0$, one has
\begin{equation}
\lim_{N\to\infty} \mathbb{P} [\max_{x\geq \kappa_{N-[\log N]^4}} h_N(x) \geq C (\log\log N)^3] = 0.
\end{equation}
By symmetry, similar estimates hold for $-h_N(x)$.
\end{lemma}

\begin{proof}
This proof is very similar to the proof of Lemma \ref{hN near edge first lemma}, with some tiny modifications. From the definition of $h_N$ in \eqref{definition of hN},  one can see that for $\kappa_j \leq x \leq \kappa_{j+1}$, one has
\begin{equation}
h_N(\kappa_j) - \sqrt{2}\pi \leq h_N(x) \leq h_N(\kappa_{j+1}) + \sqrt{2}\pi.
\end{equation}
And it further implies that
\begin{eqnarray*}
\mathbb{P} \left(\max_{x\geq \kappa_{N - [\log N]^4}} h_N(x) \geq C (\log\log N)^3 \right) &\leq& \mathbb{P} \left(\max_{j\geq N - [\log N]^4} h_N(\kappa_j) \geq C (\log\log N)^3 - \sqrt{2}\pi\right)  \\
&\leq& \sum_{j\geq N - [\log N]^4} \mathbb{P} \left(h_N(\kappa_j)\geq C (\log\log N)^3 -\sqrt{2}\pi\right).
\end{eqnarray*}
By Markov's inequality, we have, for a constant $\gamma > 0$,
\begin{equation}
\mathbb{P} \left(\max_{x\geq \kappa_{N - [\log N]^4}} h_N(x) \geq C (\log\log N)^3 \right)
\leq  \sum_{j\geq N - [\log N]^4} C' \frac{\mathbb{E} [e^{\gamma h_N(\kappa_j)}]}{e^{\gamma C (\log\log N)^3}}.
\end{equation}
We split the above sum into two parts based on the location of $\kappa_j$. Define the threshold $t_N = 1 - N^{-2/3} \log\log N$, and let
\begin{align*}
\mathcal{P}_1 = \sum_{j : \, \kappa_{N- [\log N]^4} \leq \kappa_j < t_N } C' \frac{\mathbb{E} [e^{\gamma h_N(\kappa_j)} ]}{e^{\gamma C (\log\log N)^3}}, \qquad 
\mathcal{P}_2 = \sum_{j : \, \kappa_j > t_N} C' \frac{\mathbb{E} [ e^{\gamma h_N(\kappa_j)} ]}{e^{\gamma C (\log\log N)^3}}.
\end{align*}
We now bound $\mathcal{P}_1$ and $\mathcal{P}_2$ separately. For $\mathcal{P}_1$, using \eqref{hN moment bulk right}, we have
\begin{eqnarray}
\mathcal{P}_1 \leq  \sum_{j : \, \kappa_{N- [\log N]^4} \leq \kappa_j < t_N }  C' \frac{(N F(\kappa_j))^{\gamma^2 / 2}}{e^{\gamma C (\log\log N)^3}}.
\end{eqnarray}
By letting $\gamma = \sqrt{2}$, we obtain
\begin{eqnarray}
\mathcal{P}_1 \leq \sum_{j \geq N - [\log N]^4} C' \frac{j}{e^{\sqrt{2} C (\log\log N)^3}} < C' \frac{ N[\log N]^4}{e^{\sqrt{2} C (\log\log N)^3}}.
\end{eqnarray}
Hence, we have $\mathcal{P}_1 \to 0$ as $N\to\infty$. 

For $\mathcal{P}_2$, using \eqref{hN moment edge 1}, we have
\begin{eqnarray}
\mathcal{P}_2 \leq \sum_{j : \, \kappa_j > t_N} C' \frac{e^{C(\log\log N)^{3/2}}}{e^{\gamma C (\log\log N)^3}}.
\end{eqnarray}
Since the number of such indices $j$ is of the order $ O\left((\log\log N)^{3/2} \right)$, we have 
\begin{eqnarray}
\mathcal{P}_2 \leq C' (\log\log N)^{3/2} \frac{e^{C(\log\log N)^{3/2}}}{e^{\gamma C (\log\log N)^3}},
\end{eqnarray}
which tends to $0$ as $N\to\infty$. This concludes the proof of the lemma.
\end{proof}

With this lemma in hand, a simple argument concludes our proof of soft edge rigidity.

\begin{proposition} \label{Main theorem: Eigenvalue Rigidity near Soft Edge}
For $F(x)$ defined in \eqref{distribution of eq measure}, we have
\begin{equation} \label{eq: Eigenvalue Rigidity near Soft Edge}
\lim_{N\to\infty}\mathbb{P}\left(F'(\kappa_j)|\lambda_j - \kappa_j| \leq \frac{(1+\varepsilon)\log N}{\pi N} \text{ for } j = N - [\log N]^4, \cdots, N \right) = 1.
\end{equation}
\end{proposition}

\begin{proof}
We further divide this regime into two parts. \\
\underline{Regime 1: $j = N- [\log N]^4, \cdots, N - (\log N)^{1/2}$.} We claim that in this regime, one has
\begin{equation}
\lim_{N\to\infty} \mathbb{P} (\lambda_j - \kappa_j < N^{-2/3}) = 1.
\end{equation}
The proof is almost the same as Lemma \ref{lemma: bulk estimation of lambda-kappa near 1}, with the last equation replaced by the result in Lemma \ref{hN near soft edge first lemma}. The same argument as Proposition \ref{Main theorem: Eigenvalue Rigidity on the Bulk: right} applies, hence we have
\begin{equation} \label{eq: Eigenvalue Rigidity near Soft Edge 1}
\lim_{N\to\infty}\mathbb{P}\left(F'(\kappa_j)|\lambda_j - \kappa_j| \leq \frac{(1+\varepsilon)\log N}{\pi N} \text{ for } j = N - [\log N]^4, \cdots, N-(\log N)^{1/2} \right) = 1.
\end{equation}

\noindent\underline{Regime 2: $j = N - (\log N)^{1/2}, \cdots, N$.} This is almost direct by noting that for $j$ in this regime, one has
\begin{eqnarray}
F'(\kappa_j)(\lambda_j - \kappa_j) &\leq& F'(\kappa_{N-(\log N)^{1/2}}) (1 - \kappa_{N-(\log N)^{1/2}}) \nonumber\\
&\leq & C \frac{(\log N)^{1/2}}{N},
\end{eqnarray}
for sufficiently large $N$ with probability $1$. This concludes the proof of the proposition.

\end{proof}

Combining Proposition \ref{Main theorem: Eigenvalue Rigidity on the Bulk: left}, \ref{main proposition for edge rigidity}, \ref{Main theorem: Eigenvalue Rigidity on the Bulk: right}, \ref{Main theorem: Eigenvalue Rigidity near Soft Edge}, we finally conclude the proof of Theorem \ref{main theorem}.

\appendix

\section{Several Model RH problems}

\subsection{Confluent Hypergeometric Parametrix}\label{Appendix: Section: HG model RHP}

Denote by $\Phi_{\text{HG}}(z)=\Phi_{\text{HG}}(z;\beta)$ the confluent hypergeometric parametrix, where $\beta$ is a parameter. It satisfies the following RH problem.

\begin{rhp} \label{Appendix: HG model RHP}
\hfill

\begin{itemize}
  \item[(a)]   $\Phi_{\text{HG}}(z)$ is analytic in $\mathbb{C}\setminus \{\cup^6_{j=1}\widehat\Sigma_j\cup\{0\}\}$, where the contours $\widehat\Sigma_j$, $j=1,\ldots,6,$ are indicated in Figure. \ref{fig:jumps-Phi-C}.

  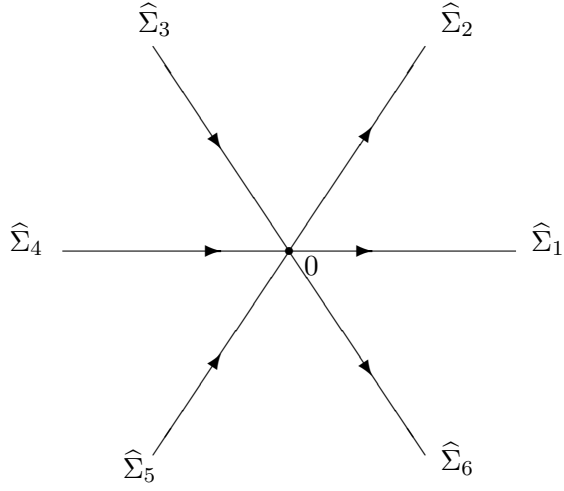
\begin{figure}[h]
\centering
   \setlength{\unitlength}{1truemm}
   \begin{picture}(100,70)(-5,2)
       \put(40,40){\line(-2,-3){18}}
       \put(40,40){\line(-2,3){18}}
       \put(40,40){\line(-1,0){30}}
       \put(40,40){\line(1,0){30}}
  \put(40,40){\line(2,-3){18}}
    \put(40,40){\line(2,3){18}}

       \put(30,55){\thicklines\vector(2,-3){1}}
       \put(30,40){\thicklines\vector(1,0){1}}
       \put(50,40){\thicklines\vector(1,0){1}}
       \put(30,25){\thicklines\vector(2,3){1}}
      \put(50,25){\thicklines\vector(2,-3){1}}
       \put(50,55){\thicklines\vector(2,3){1}}


       \put(42,36.9){$0$}
         \put(72,40){$\widehat \Sigma_1$}
           \put(60,69){$\widehat \Sigma_2$}
             \put(20,69){$\widehat \Sigma_3$}
              \put(3,40){$\widehat \Sigma_4$}
       \put(18,10){$\widehat \Sigma_5$}
          \put(60,11){$ \widehat \Sigma_6$}

%

       \put(40,40){\thicklines\circle*{1}}
\end{picture}
   \caption{The jump contours for the RH problem for $\Phi_{\text{HG}}$.}
   \label{fig:jumps-Phi-C}

\end{figure}

  \item[(b)] $\Phi_{\text{HG}}$ satisfies the following jump condition:
  \begin{equation}\label{HJumps}
  \Phi_{\text{HG}, +}(z)=\Phi_{\text{HG}, -} \widehat J_i(z), \quad z \in \widehat\Sigma_i,\quad j=1,\ldots,6,
  \end{equation}
  where
  \begin{equation*}
 \widehat J_1(z) = \begin{pmatrix}
    0 &   e^{-\beta \pi i} \\
    -  e^{\beta \pi i} &  0
    \end{pmatrix}, \qquad \widehat J_2(z) = \begin{pmatrix}
    1 & 0 \\
    e^{ \beta \pi i } & 1
    \end{pmatrix}, \qquad
    \widehat J_3(z) = \begin{pmatrix}
    1 & 0 \\
    e^{ -\beta\pi i} & 1
    \end{pmatrix},                                                         
  \end{equation*}
  \begin{equation*}
  \widehat J_4(z) = \begin{pmatrix}
    0 &   e^{\beta\pi i} \\
     -  e^{-\beta\pi i} &  0
     \end{pmatrix}, \qquad
      \widehat J_5(z) = \begin{pmatrix}
     1 & 0 \\
     e^{- \beta\pi i} & 1
     \end{pmatrix},\qquad
     \widehat J_6(z) = \begin{pmatrix}
   1 & 0 \\
   e^{\beta\pi i} & 1
   \end{pmatrix}.
  \end{equation*}

  \item[(c)] $\Phi_{\text{HG}}$ satisfies the following asymptotic behavior at infinity:
  \begin{equation} \label{Appendix: model RHP HG asymptotics}
\Phi_{\text{HG}}(z) = e^{-\frac{\pi i}{2}\beta} (I + O(z^{-1})) (z)^{-\beta\sigma_3} e^{-i\frac{z}{2}\sigma_3} \tilde{\chi}(z), \qquad z\to\infty
\end{equation}
where $\tilde{\chi}(z)$ is given by
\begin{eqnarray}
\tilde{\chi}(z) = \left\{\begin{array}{ll}
e^{\pi i \beta\sigma_3}, & \arg z\in (0, \pi)\\
\left(
  \begin{array}{cc}
    0 & -1 \\
    1 & 0 \\
  \end{array}
\right), & \arg z \in (\pi, 2\pi).
\end{array}
\right.
\end{eqnarray}
  
\item[(d)] As $z\to 0$, we have $\Phi_{\text{HG}}(z)=\Boh(\log |z|)$.

\end{itemize}

\end{rhp}

From \cite{IK}, it follows that the above RH problem can be solved explicitly in the following way. For $z$ belonging to the region bounded by the rays $\widehat \Sigma_1$ and $\widehat \Sigma_2$,
\begin{equation}\label{Hsolution}
\Phi_{\text{HG}}(z)=C_1\left(\begin{array}{ll}
\psi(\beta,1,e^{\frac{\pi i}{2}}z)e^{2 \beta \pi i}e^{-\frac{iz}{2}}&-\frac{\Gamma(1-\beta)}{\Gamma(\beta)}\psi(1-\beta,1,e^{-\frac{\pi i}{2}}z)e^{\beta \pi i}e^{\frac{iz}{2}}\\
-\frac{\Gamma(1+\beta)}{\Gamma(-\beta)}\psi(1+\beta,1,e^{\frac{\pi i}{2}}z)e^{\beta \pi i}e^{-\frac{iz}{2}}
&\psi(-\beta,1,e^{-\frac{\pi i}{2}}z)e^{\frac{iz}{2}}\end{array}\right),
\end{equation}
where the confluent hypergeometric function $\psi(a,b;z)$ is the unique solution to the Kummer's equation
\begin{equation} \label{Kummer-equation}
z\frac{\ud^2y}{\ud z^2}+(b-z)\frac{\ud y}{\ud z}-ay=0
\end{equation}
satisfying the boundary condition $\psi(a,b,z)\sim  z^{-a}$ as $z\to \infty$ and $-\frac{3\pi }{2} < \arg z < \frac{3\pi}{2}$;
see \cite[Chapter 13]{DLMF}.  The branches of the multi-valued functions are chosen such that $-\frac{\pi}{2}<\arg z<\frac{3\pi}{2}$ and
$$
C_1=
\begin{pmatrix}
e^{-\frac32 \beta \pi i} & 0
\\
0 & e^{\frac12 \beta \pi i}
\end{pmatrix}
$$
is a constant matrix. The explicit formula of $\Phi_{\text{HG}}(z)$ in the other sectors is then determined by using the jump condition \eqref{HJumps}. 

From \cite[Lemma C.1]{CFL2021}, one can show that for $\Re\beta = 0$, taking the limit where $z$ approaches $0$ with $\arg z\in (\pi/4,\pi/2)$, one has
\begin{equation}\label{Appendix: HG specific value}
\lim_{z\to 0} \left(\Phi_{HG}(z)\right)_{1,1} = \Gamma (1-\beta),\qquad \lim_{z\to 0} \left(\Phi_{HG}(z)\right)_{2,1} = \Gamma (1-\beta)
\end{equation}
and
\begin{equation} \label{Appendix: HG Lemma}
\lim_{z\to 0}\left(\Phi_{\text{HG}}^{-1}(z)\frac{d}{dz}\Phi_{\text{HG}}(z)\right)_{2,1}={ -\frac{2\pi  \beta}{e^{\pi i \beta}-e^{-\pi i \beta}}}.
\end{equation}

\subsection{A Modified Painlev\'e V Model RH Problem.} \label{Appendix: Section: PV model RHP}

The modified Painlev\'e V parametrix $\widehat{\Phi}_{PV}(z; s)$ with $s$ being a parameter is a solution to the following RH problem. 

\begin{rhp} \label{Appendix: PV RHP}
\hfill
\begin{itemize}
    \item[(a)] $\widehat{\Phi}_{PV}:\mathbb C\setminus \{ \cup_{j=1}^7 \widehat{\Gamma}_j \cup \{0,1\}\} \to \mathbb C^{2\times 2}$ is analytic, where $\widehat{\Gamma}_j$, $j=1,2,...,7$ are depicted in Figure \ref{fig:hatPsi}.
    \item[(b)] $\widehat{\Phi}_{PV}$ satisfies the jump conditions
    \begin{equation}\label{jump hatPsi}\widehat{\Phi}_{PV, +}(z)=\widehat{\Phi}_{PV, -}(z)\widehat J_k,\qquad
    z\in\widehat\Gamma_k,\end{equation}  where
                \begin{align*}
                &\widehat J_1=\begin{pmatrix}1&0\\e^{-{\sqrt{2}\pi}(\gamma_1+\gamma_2)}&1\end{pmatrix},
                &&\widehat J_2=\begin{pmatrix}1&0\\e^{-{\sqrt{2}\pi}(\gamma_1+\gamma_2)}&1\end{pmatrix},\\
                &\widehat J_3=\begin{pmatrix}1&0\\1&1\end{pmatrix},
                &&\widehat J_4=\begin{pmatrix}1&0\\ 1&1\end{pmatrix},\\
                &\widehat J_5=\begin{pmatrix}1&e^{{\sqrt{2}\pi}\gamma_2}\\0&1\end{pmatrix},
&&\widehat J_6=\begin{pmatrix}0&e^{{\sqrt{2}\pi}\gamma_1+{\sqrt{2}\pi}\gamma_2}\\ -e^{-{\sqrt{2}\pi}\gamma_1-{\sqrt{2}\pi}\gamma_2}&0\end{pmatrix},\\
                &\widehat J_7=\begin{pmatrix}0&1\\-1&0\end{pmatrix},                
                \end{align*}
    \item[(c)] $\Phi_{\text{PV}}$ satisfies the following asymptotic behavior at infinity:
    \begin{equation}\label{Psi ashat}
    \widehat{\Phi}_{PV}(z)=\left(I+\frac{\widehat{\Phi}_1}{z}+\frac{\widehat{\Phi}_2}{z^2}+O(z^{-3})\right)
    \widehat{\Phi}^\infty(z)e^{\pm \frac{s}{2}z\sigma_3} \qquad  \mbox{ as $z\to \infty$ \text{ and } $\pm \mathrm{Im}z>0$,}
    \end{equation}
    where
    $\widehat{\Phi}^\infty$ is defined in \eqref{hatpsiinfty}.
                \item[(d)] As $z \to x\in\{0,1\}$, we have
                \begin{equation*}
                \widehat{\Phi}_{PV}(z;s)=O(\log(z-x)).
                \end{equation*}
    \end{itemize}

\begin{figure}[h]
\centering
    \setlength{\unitlength}{0.8truemm}
    \begin{picture}(75,65)(5,10)
    \put(45,50){\thicklines\circle*{.8}}
    \put(60,50){\thicklines\circle*{.8}}
    \put(30,50){\thicklines\circle*{.8}}
    \put(55,50){\thicklines\vector(1,0){.0001}}
    \put(40,50){\thicklines\vector(1,0){.0001}}
    \put(15,50){\thicklines\vector(1,0){.0001}}
    \put(75,50){\thicklines\vector(1,0){.0001}}
    \put(0,50){\line(1,0){90}}
    \put(60,50){\line(1,1){25}}
    \put(30,50){\line(-1,1){25}}
    \put(60,50){\line(1,-1){25}}
    \put(30,50){\line(-1,-1){25}}
    \put(75,65){\thicklines\vector(1,1){.0001}}
    \put(15,65){\thicklines\vector(1,-1){.0001}}
    \put(75,35){\thicklines\vector(1,-1){.0001}}
    \put(15,35){\thicklines\vector(1,1){.0001}}

    \put(75,75){\small $\widehat{\Gamma}_4$}
    \put(15,75){\small $\widehat{\Gamma}_1$}
    \put(75,20){\small $\widehat{\Gamma}_3$}
    \put(15,20){\small $\widehat{\Gamma}_2$}
    \put(45,54){\small $\widehat{\Gamma}_5$}
    \put(15,54){\small $\widehat{\Gamma}_6$}
    \put(75,54){\small $\widehat{\Gamma}_7$}

    \put(30,53){\small $0$}
    \put(59,53){\small $1$}

    \end{picture}

    \caption{The jump contour for $\widehat{\Phi}_{PV}$.}
    \label{fig:hatPsi}
\end{figure}
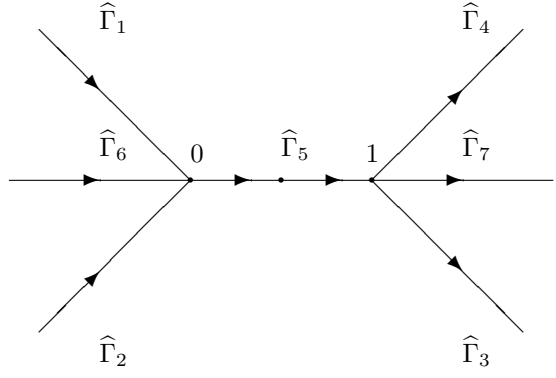
    
\end{rhp}

It is shown in \cite[Lemma 6.1]{CFL2021} that the above RH problem is uniquely solvable when $\gamma_1, \gamma_2 \in \mathbb{R}$ and $s \in -i \mathbb{R}_+$. The solution can be constructed by using the $\psi$-functions associated with the Painlev\'e V equation; see also \cite[Sec. 3]{Cla:Kra2015}.

\subsection{A Model RH Problem for Soft Edge}\label{Appendix:subsec:modelrhp-softedge}

A new model RH problem for studying the case where the jump point $x$ is close to the soft edge $1$ is constructed in \cite[Sec. 6.2]{CFL2021}. It is stated below with only a slight difference in the asymptotics of $\Phi_{\text{soft}}$.

\begin{figure}[htbp]

\centering

\tikzset{every picture/.style={line width=0.75pt}} 

\begin{tikzpicture}[x=0.75pt,y=0.75pt,yscale=-1,xscale=1]

\draw    (100,123) -- (194.78,122.67) -- (289.56,122.33) -- (384.33,122) ;
\draw [shift={(152.39,122.82)}, rotate = 179.8] [fill={rgb, 255:red, 0; green, 0; blue, 0 }  ][line width=0.08]  [draw opacity=0] (8.93,-4.29) -- (0,0) -- (8.93,4.29) -- cycle    ;
\draw [shift={(247.17,122.48)}, rotate = 179.8] [fill={rgb, 255:red, 0; green, 0; blue, 0 }  ][line width=0.08]  [draw opacity=0] (8.93,-4.29) -- (0,0) -- (8.93,4.29) -- cycle    ;
\draw [shift={(341.94,122.15)}, rotate = 179.8] [fill={rgb, 255:red, 0; green, 0; blue, 0 }  ][line width=0.08]  [draw opacity=0] (8.93,-4.29) -- (0,0) -- (8.93,4.29) -- cycle    ;
\draw    (194.78,122.67) .. controls (221.33,84) and (266.33,83) .. (294.78,122.67) ;
\draw [shift={(250.05,93.61)}, rotate = 183.01] [fill={rgb, 255:red, 0; green, 0; blue, 0 }  ][line width=0.08]  [draw opacity=0] (8.93,-4.29) -- (0,0) -- (8.93,4.29) -- cycle    ;
\draw    (194.78,122.67) .. controls (218.33,162) and (270.33,162) .. (294.78,122.67) ;
\draw [shift={(250.14,151.87)}, rotate = 177.33] [fill={rgb, 255:red, 0; green, 0; blue, 0 }  ][line width=0.08]  [draw opacity=0] (8.93,-4.29) -- (0,0) -- (8.93,4.29) -- cycle    ;
\draw    (116.33,45) -- (194.78,122.67) ;
\draw [shift={(159.11,87.35)}, rotate = 224.71] [fill={rgb, 255:red, 0; green, 0; blue, 0 }  ][line width=0.08]  [draw opacity=0] (8.93,-4.29) -- (0,0) -- (8.93,4.29) -- cycle    ;
\draw    (117.33,198) -- (194.78,122.67) ;
\draw [shift={(159.64,156.85)}, rotate = 135.79] [fill={rgb, 255:red, 0; green, 0; blue, 0 }  ][line width=0.08]  [draw opacity=0] (8.93,-4.29) -- (0,0) -- (8.93,4.29) -- cycle    ;

\draw (178,133) node [anchor=north west][inner sep=0.75pt]   [align=left] {$\displaystyle -1$};
\draw (290,128) node [anchor=north west][inner sep=0.75pt]   [align=left] {$\displaystyle 0$};
\draw (123,199) node [anchor=north west][inner sep=0.75pt]   [align=left] {$\displaystyle \Sigma _{\Phi ,1}$};
\draw (123,24) node [anchor=north west][inner sep=0.75pt]   [align=left] {$\displaystyle \Sigma _{\Phi ,1}$};
\draw (235,170) node [anchor=north west][inner sep=0.75pt]   [align=left] {$\displaystyle \Sigma _{\Phi ,2}$};
\draw (242,53) node [anchor=north west][inner sep=0.75pt]   [align=left] {$\displaystyle \Sigma _{\Phi ,2}$};

\end{tikzpicture}
\caption{The jump contour $\Sigma_{\Phi}$}\label{ContourPhi}
\end{figure}
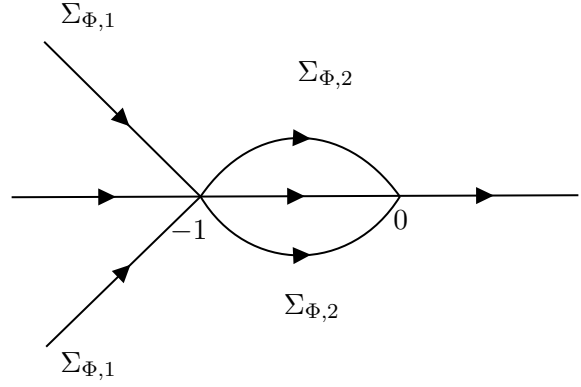

\begin{rhp}\label{model-rhp-soft-edge}
\hfill

\begin{itemize}
\item[(a)] $\Phi_{\text{soft}} = \Phi_{\text{soft}}(\cdot;u)$ is analytic on $\mathbb C\setminus (\mathbb R \cup \Sigma_{\Phi,1}\cup\Sigma_{\Phi,2})$. The contours $\Sigma_{\Phi,1}$ and $\Sigma_{\Phi,2}$ are as in Figure \ref{ContourPhi}: they consist of straight lines near $0$, and near $-1$ they will be specified below.
\item[(b)] On $\mathbb R \cup \Sigma_{\Phi,1}\cup\Sigma_{\Phi,2} \setminus \{-1,0\}$, $\Phi_{\text{soft}}$ has the following jumps:
\begin{equation}
\begin{aligned}
\Phi_{\text{soft}, +}(\lambda)&=\Phi_{\text{soft}, -}(\lambda) \begin{pmatrix}
1&0\\e^{-{\sqrt{2}\pi}\gamma}e^{\frac{4}{3}(\lambda u)^{3/2}}&1
\end{pmatrix}	&& \textrm{for } \lambda\in \Sigma_{\Phi,1},
\\
\Phi_{\text{soft}, +}(\lambda)&=\Phi_{\text{soft}, -}(\lambda) \begin{pmatrix}
1&0\\e^{\frac{4}{3}(\lambda u)^{3/2}}&1
\end{pmatrix}	&& \textrm{for } \lambda\in \Sigma_{\Phi,2},
\\
\Phi_{\text{soft}, +}(\lambda)&=\Phi_{\text{soft}, -}(\lambda) \begin{pmatrix}
0&e^{{\sqrt{2}\pi}\gamma}\\-e^{-{\sqrt{2}\pi}\gamma}&0
\end{pmatrix}	&& \textrm{for } \lambda\in (-\infty,-1),
\\
\Phi_{\text{soft}, +}(\lambda)&=\Phi_{\text{soft}, -}(\lambda) \begin{pmatrix}
0&1\\-1&0
\end{pmatrix}	&& \textrm{for } \lambda\in (-1,0),
\\
\Phi_{\text{soft}, +}(\lambda)&=\Phi_{\text{soft}, -}(\lambda) \begin{pmatrix}
1&e^{-\frac{4}{3}(\lambda u)^{3/2}}\\0&1
\end{pmatrix}	&& \textrm{for } \lambda\in (0,\infty),\label{eq:jumpPhi}
\end{aligned}
\end{equation}
where principal branches are chosen, and where $u>0$ and { $\gamma \in \mathbb R$} are parameters.
\item[(c)] As $\lambda\to \infty$,
\begin{equation}\label{eq:Phiasy}
\Phi_{\text{soft}}(\lambda)={  (I+O(\lambda^{-1}))\begin{pmatrix}
1&0\\ i{\sqrt{2}}\gamma &1
\end{pmatrix}\lambda^{-\frac{1}{4}\sigma_3} B
e^{-{\sqrt{2}\pi}\frac{\gamma}{2}\sigma_3},}
\end{equation}
where principal branches are chosen, and $B=\frac{1}{\sqrt 2}\begin{pmatrix}
1&i\\i&1
\end{pmatrix}$.
\item[(d)] As $\lambda\to 0$, $\Phi_{\text{soft}}(\lambda)$ remains bounded and as $\lambda\to -1$, $\Phi_{\text{soft}}(\lambda)=\mathcal O(\log(\lambda+1))$.
\end{itemize}
\end{rhp}

As discussed in \cite[Sec. 6.2]{CFL2021}, we do not need a specific solution to this RH problem. The existence and asymptotics of the solution $\Phi_{\text{soft}}$ has been established in \cite{CFL2021}, and they are sufficient for our study.

\subsection{A Model RH Problem for Hard Edge}\label{Appendix:subsec:modelrhp-hardedge}

A new model RH problem for studying the case where the jump point $x$ is close to the hard edge is constructed in \cite[Sec. 3.5.1]{Dai-Lu-JUE}. 

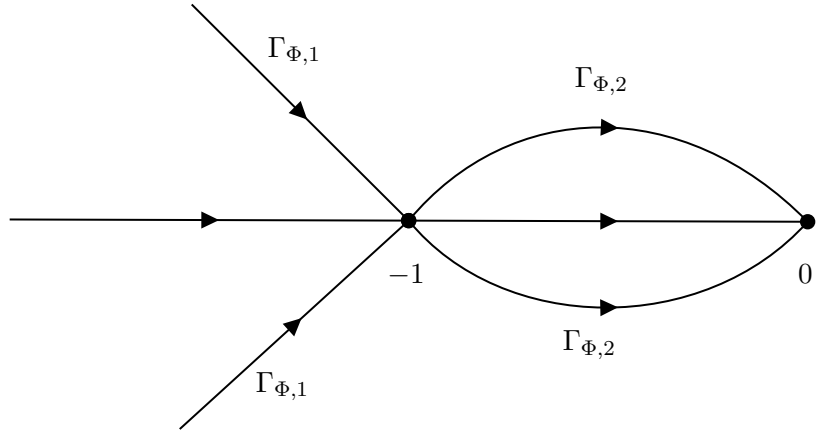
\begin{figure}[htbp]

\centering

\tikzset{every picture/.style={line width=0.75pt}} 

\begin{tikzpicture}[x=0.75pt,y=0.75pt,yscale=-1,xscale=1]

\draw    (100,122) -- (300.17,122.57) ;
\draw [shift={(300.17,122.57)}, rotate = 0.16] [color={rgb, 255:red, 0; green, 0; blue, 0 }  ][fill={rgb, 255:red, 0; green, 0; blue, 0 }  ][line width=0.75]      (0, 0) circle [x radius= 3.35, y radius= 3.35]   ;
\draw [shift={(205.08,122.3)}, rotate = 180.16] [fill={rgb, 255:red, 0; green, 0; blue, 0 }  ][line width=0.08]  [draw opacity=0] (8.93,-4.29) -- (0,0) -- (8.93,4.29) -- cycle    ;
\draw    (300.17,122.57) -- (500.33,123.14) ;
\draw [shift={(500.33,123.14)}, rotate = 0.16] [color={rgb, 255:red, 0; green, 0; blue, 0 }  ][fill={rgb, 255:red, 0; green, 0; blue, 0 }  ][line width=0.75]      (0, 0) circle [x radius= 3.35, y radius= 3.35]   ;
\draw [shift={(405.25,122.87)}, rotate = 180.16] [fill={rgb, 255:red, 0; green, 0; blue, 0 }  ][line width=0.08]  [draw opacity=0] (8.93,-4.29) -- (0,0) -- (8.93,4.29) -- cycle    ;
\draw [shift={(300.17,122.57)}, rotate = 0.16] [color={rgb, 255:red, 0; green, 0; blue, 0 }  ][fill={rgb, 255:red, 0; green, 0; blue, 0 }  ][line width=0.75]      (0, 0) circle [x radius= 3.35, y radius= 3.35]   ;
\draw    (191.33,14.14) -- (300.17,122.57) ;
\draw [shift={(249.29,71.88)}, rotate = 224.89] [fill={rgb, 255:red, 0; green, 0; blue, 0 }  ][line width=0.08]  [draw opacity=0] (8.93,-4.29) -- (0,0) -- (8.93,4.29) -- cycle    ;
\draw    (185.33,227.14) -- (300.17,122.57) ;
\draw [shift={(246.45,171.49)}, rotate = 137.68] [fill={rgb, 255:red, 0; green, 0; blue, 0 }  ][line width=0.08]  [draw opacity=0] (8.93,-4.29) -- (0,0) -- (8.93,4.29) -- cycle    ;
\draw    (300.17,122.57) .. controls (351.33,61.14) and (437.33,59.14) .. (500.33,123.14) ;
\draw [shift={(405.5,76.15)}, rotate = 182.72] [fill={rgb, 255:red, 0; green, 0; blue, 0 }  ][line width=0.08]  [draw opacity=0] (8.93,-4.29) -- (0,0) -- (8.93,4.29) -- cycle    ;
\draw    (300.17,122.57) .. controls (347.33,180.14) and (448.33,181.14) .. (500.33,123.14) ;
\draw [shift={(405.22,166.07)}, rotate = 178.76] [fill={rgb, 255:red, 0; green, 0; blue, 0 }  ][line width=0.08]  [draw opacity=0] (8.93,-4.29) -- (0,0) -- (8.93,4.29) -- cycle    ;

\draw (288,143) node [anchor=north west][inner sep=0.75pt]   [align=left] {$\displaystyle -1$};
\draw (494,143) node [anchor=north west][inner sep=0.75pt]   [align=left] {$\displaystyle 0$};
\draw (228,29) node [anchor=north west][inner sep=0.75pt]   [align=left] {$\displaystyle \Gamma _{\Phi ,1}$};
\draw (222,197) node [anchor=north west][inner sep=0.75pt]   [align=left] {$\displaystyle \Gamma _{\Phi ,1}$};
\draw (382,44) node [anchor=north west][inner sep=0.75pt]   [align=left] {$\displaystyle \Gamma _{\Phi ,2}$};
\draw (376,176) node [anchor=north west][inner sep=0.75pt]   [align=left] {$\displaystyle \Gamma _{\Phi ,2}$};

\end{tikzpicture}

\caption{The jump contours of the RH problem for $\Phi$} \label{Figure: Model RHP}

\end{figure}

This model RH problem is similar to that in \cite[Sec. 6.2]{CFL2021}, but with a slightly different contour. This difference arises because the endpoint $z=1$ is a hard edge in \cite{Dai-Lu-JUE}, in contrast to the soft edge in \cite{CFL2021}.

\begin{rhp} \label{model rhp for Phi}
\hfill
\begin{itemize}
\item[(a)] $\Phi = \Phi(\lambda; u)$ is analytic on $\mathbb{C}\setminus ((-\infty, 0]\cup \Gamma_{\Phi, 1}\cup \Gamma_{\Phi, 2})$; see Figure \ref{Figure: Model RHP}.

\item[(b)] On $(-\infty, 0]\cup \Gamma_{\Phi, 1}\cup \Gamma_{\Phi, 2}\setminus \{-1, 0\}$, $\Phi$ satisfies the following jump conditions:
    \begin{eqnarray} \label{Model-jump-Phi}
    \begin{array}{ll}
    \Phi_+(\lambda) = 
    \left\{
    \begin{array}{ll}
    \Phi_-(\lambda)\left(\begin{array}{ll} 1 & 0\\ e^{-\sqrt{2}\pi\gamma}e^{-4(\lambda u)^{1/2}}e^{\pm \alpha\pi i} & 1
     \end{array}\right) & \text{ for } \lambda\in\Gamma_{\Phi, 1}, \text{ and } \lambda \in \mathbb{C}^{\pm}, \\
     \Phi_-(\lambda)\left(\begin{array}{ll} 1 & 0\\ e^{-4(\lambda u)^{1/2}}e^{\pm \alpha\pi i} & 1
     \end{array}\right) & \text{ for } \lambda\in\Gamma_{\Phi, 2}, \text{ and } \lambda \in \mathbb{C}^{\pm}, \\
     \Phi_-(\lambda)\left(\begin{array}{ll} 0 & e^{\sqrt{2}\pi\gamma}\\ -e^{-\sqrt{2}\pi\gamma} & 0
     \end{array}\right) & \text{ for } \lambda\in (-\infty, -1),\\
     \Phi_-(\lambda)\left(\begin{array}{ll} 0 & 1\\ -1 & 0
     \end{array}\right) & \text{ for } \lambda \in (-1, 0),
     \end{array}
     \right.
    \end{array}
    \end{eqnarray}
   where principal branches are chosen in the square roots, and $u>0$ and $\gamma\in\mathbb{R}$ are parameters.

\item[(c)] As $\lambda\to\infty$,
\begin{equation}\label{Phiasym}
\Phi(\lambda) = (I+O(\lambda^{-1}))\left(\begin{array}{cc} 1 & 0\\ i\sqrt{2}\gamma & 1\end{array}\right)\lambda^{-\frac{1}{4}\sigma_3} B e^{-\sqrt{2}\pi \frac{\gamma}{2}\sigma_3},
\end{equation}
where the principal branches are chosen, and
\begin{equation}\label{A}
B = \frac{1}{\sqrt 2}\left(\begin{array}{cc} 1 & i\\ i & 1\end{array}\right).
\end{equation}

\item[(d)] As $\lambda\to -1$, $\Phi(\lambda) = O(\log |\lambda+1|)$.

\item[(e)] For $\alpha<0$, $\Phi(\lambda)$ has the following behaviour:
\begin{eqnarray}
\Phi(\lambda) = O\left(
                   \begin{array}{cc}
                     |\lambda|^{\alpha / 2} & |\lambda|^{\alpha / 2} \\
                     |\lambda|^{\alpha / 2} & |\lambda|^{\alpha / 2} \\
                   \end{array}
                 \right), \text{ as } \lambda\to 0.
\end{eqnarray}
For $\alpha=0$, $\Phi(\lambda)$ has the following behaviour:
\begin{eqnarray}
\Phi(\lambda) = O\left(
                   \begin{array}{cc}
                     \log |\lambda| & \log |\lambda| \\
                     \log |\lambda| & \log |\lambda| \\
                   \end{array}
                 \right), \text{ as } \lambda\to 0.
\end{eqnarray}
For $\alpha>0$, $\Phi(\lambda)$ has the following behaviour:
\begin{eqnarray}
\Phi(\lambda) = \left\{\begin{array}{ll}
O\left(
   \begin{array}{cc}
     |\lambda|^{\alpha/2} & |\lambda|^{-\alpha/2} \\
     |\lambda|^{\alpha/2} & |\lambda|^{-\alpha/2} \\
   \end{array}
 \right), & \text{as } \lambda\to 0 \text{ outside the lens,}\\
O\left(
   \begin{array}{cc}
     |\lambda|^{-\alpha/2} & |\lambda|^{-\alpha/2} \\
     |\lambda|^{-\alpha/2} & |\lambda|^{-\alpha/2} \\
   \end{array}
 \right), & \text{as } \lambda\to 0 \text{ inside the lens,}
\end{array}
\right.
\end{eqnarray}
\end{itemize}
\end{rhp}

Similarly, to obtain the asymptotics of the related Hankel determinants, an explicit expression for $\Phi(\lambda; u)$ is unnecessary; only its existence for sufficiently large $u$ is required. This existence can be established by analyzing the asymptotic behavior of the RH problem as $u \to +\infty$.

\paragraph*{Acknowledgements}
The authors would like to express their sincere gratitude to Prof. Dan Dai for his invaluable and constructive guidance and discussions.

\end{document}